\documentclass[11pt,a4paper,onecolumn,accepted=2023-07-23]{quantumarticle}
\pdfoutput=1 
\usepackage[T1]{fontenc}

\usepackage{tikz}
\usetikzlibrary{positioning, calc}
\usetikzlibrary{calc}
\usetikzlibrary{arrows}
\usepackage{tikz-3dplot}
\usepackage{graphicx}
\usepackage{mdwlist}
\usepackage{color}
\usepackage{thmtools}
\usepackage{amssymb}
\usepackage{physics}
\usepackage{amsmath}
\usepackage{array}
\usepackage{url,hyperref}
\usepackage[numbers]{natbib}
\usetikzlibrary{fadings}
\usetikzlibrary{decorations.pathmorphing}
\usepackage{mathrsfs}
\usepackage{authblk}
\usepackage{amsmath}

\usepackage{cleveref}
\DeclareMathOperator{\spn}{span}
\DeclareMathOperator{\size}{size}

\usetikzlibrary{decorations.pathreplacing}
\tikzset{snake it/.style={decorate, decoration=snake}}

\usetikzlibrary{decorations.pathreplacing,decorations.markings}

\tikzset{
    >=stealth',
    punkt/.style={
           rectangle,
           rounded corners,
           draw=black, very thick,
           text width=6.5em,
           minimum height=2em,
           text centered},
    pil/.style={
           ->,
           thick,
           shorten <=2pt,
           shorten >=2pt,},
  on each segment/.style={
    decorate,
    decoration={
      show path construction,
      moveto code={},
      lineto code={
        \path [#1]
        (\tikzinputsegmentfirst) -- (\tikzinputsegmentlast);
      },
      curveto code={
        \path [#1] (\tikzinputsegmentfirst)
        .. controls
        (\tikzinputsegmentsupporta) and (\tikzinputsegmentsupportb)
        ..
        (\tikzinputsegmentlast);
      },
      closepath code={
        \path [#1]
        (\tikzinputsegmentfirst) -- (\tikzinputsegmentlast);
      },
    },
  },
  mid arrow/.style={postaction={decorate,decoration={
        markings,
        mark=at position .5 with {\arrow[#1]{stealth'}}
      }}}
}

\usepackage{hyperref}
\usepackage{slashed}

\usepackage{float}		
\usepackage{graphicx}		
\usepackage{subcaption}
\usepackage{amsmath}

\usepackage{bbold}

\newcommand{\dt}{\text{depth}_Q}
\newcommand{\DTq}{\text{DT}^Q}
\newcommand{\modkl}{\text{Mod}_p\text{L}}

\newcommand{\spacet}{\text{SPACE}_{(2)}}

\newcommand{\PSP}{\text{PSP}}
\newcommand{\NL}{\text{NL}}

\renewcommand{\L}{\text{L}}
\renewcommand{\P}{\text{P}}

\newtheorem{theorem}{Theorem}

\newtheorem{definition}[theorem]{Definition}

\newtheorem{lemma}[theorem]{Lemma}

\newenvironment{proof}[1][Proof]{\noindent\textbf{#1.}}{\ \rule{0.5em}{0.5em}}

\begin{document} 

\title{Code-routing: a new attack on position verification}

\author[1]{Joy Cree}
\email{scree@stanford.edu}
\orcid{0000-0003-2283-3903}

\author[1]{Alex May}
\email{alexmay2@stanford.edu}
\orcid{0000-0002-4030-5410}

\affiliation[1]{Stanford University}

\begin{abstract}
The cryptographic task of position verification attempts to verify one party's location in spacetime by exploiting constraints on quantum information and relativistic causality.
A popular verification scheme known as $f$-routing involves requiring the prover to redirect a quantum system based on the value of a Boolean function $f$. 
Cheating strategies for the $f$-routing scheme require the prover use pre-shared entanglement, and security of the scheme rests on assumptions about how much entanglement a prover can manipulate.
Here, we give a new cheating strategy in which the quantum system is encoded into a secret-sharing scheme, and the authorization structure of the secret-sharing scheme is exploited to direct the system appropriately.
This strategy completes the $f$-routing task using $O(SP_p(f))$ EPR pairs, where $SP_p(f)$ is the minimal size of a span program over the field $\mathbb{Z}_p$ computing $f$.
This shows we can efficiently attack $f$-routing schemes whenever $f$ is in the complexity class $\modkl$, after allowing for local pre-processing. 
The best earlier construction achieved the class $\L$, which is believed to be strictly inside of $\modkl$. 
We also show that the size of a quantum secret sharing scheme with indicator function $f_I$ upper bounds entanglement cost of $f$-routing on the function $f_I$.
\end{abstract}

\vfill

\maketitle

\pagebreak

\tableofcontents

\section{Introduction}

\begin{figure*}
\begin{center}
\begin{subfigure}[b]{.8\textwidth}
\begin{center}
\begin{tikzpicture}[scale=0.55]
    
    \draw[fill=gray,opacity=0.4] (-1.5,1) -- (1.5,1) -- (1.5,7) -- (-1.5,7) -- (-1.5,1);
    
    \draw[->] (-7,-1) -- (-7,0);
    \node[above] at (-7,0) {$t$};
    \draw[->] (-7,-1) -- (-6,-1);
    \node[right] at (-6,-1) {$x$};
    
    \draw[->] (-4,-1) -- (-4,-0.1);
    \node[below] at (-4,-1) {$A_0$};
    
    \draw[->] (4,-1) -- (4,-0.1);
    \node[below] at (4,-1) {$A_1$};
    
    \node[left] at (-4,0) {$c_0$};
    \draw[fill=black] (-4,0) circle (0.15);
    
    \node[right] at (4,0) {$c_1$};
    \draw[fill=black] (4,0) circle (0.15);
    
    \draw[->] (4,8) -- (4,9);
    \node[above] at (4,9) {$B_1$};
    \draw[->] (-4,8) -- (-4,9);
    \node[above] at (-4,9) {$B_0$};

    \node[right] at (4,8) {$r_1$};
    \draw[fill=blue] (4,8) circle (0.15);

    \node[left] at (-4,8) {$r_0$};
    \draw[fill=blue] (-4,8) circle (0.15);
    
    \node[below] at (0,-0.53) {$ $};
    
    \end{tikzpicture}
\end{center}
\caption{}
\label{subfig:PVset-up}
\end{subfigure}
\begin{subfigure}[b]{.45\textwidth}
\begin{center}
\begin{tikzpicture}[scale=0.55]
    
    \draw[fill=gray,opacity=0.4] (-1.5,1) -- (1.5,1) -- (1.5,7) -- (-1.5,7) -- (-1.5,1);
    
    \draw[postaction={on each segment={mid arrow}}] (-4,0) -- (0,4);
    \draw[postaction={on each segment={mid arrow}}] (4,0) -- (0,4);
    \draw[postaction={on each segment={mid arrow}}] (0,4) -- (4,8);
    \draw[postaction={on each segment={mid arrow}}] (0,4) -- (-4,8);
    
    \node[below left] at (-4,0) {$c_0$};
    \draw[fill=black] (-4,0) circle (0.15);

    \node[below right] at (4,0) {$c_1$};
    \draw[fill=black] (4,0) circle (0.15);

    \node[above right] at (4,8) {$r_1$};
    \draw[fill=blue] (4,8) circle (0.15);

    \node[above left] at (-4,8) {$r_0$};
    \draw[fill=blue] (-4,8) circle (0.15);
    
    \draw[fill=yellow] (0,4) circle (0.30);
    
    \node[below] at (0,-0.53) {$ $};
    
    \node[above left] at (-3,1) {$A_0$};
    \node[above right] at (3,1) {$A_1$};
    
    \node[below left] at (-3,7) {$B_0$};
    \node[below right] at (3,7) {$B_1$};
    
    \end{tikzpicture}
\end{center}
\caption{}
\label{subfig:localschematic}
\end{subfigure}
\hfill
\begin{subfigure}[b]{.45\textwidth}
\begin{center}
\begin{tikzpicture}[scale=0.55]

    \draw[fill=gray,opacity=0.4] (-1.5,1) -- (1.5,1) -- (1.5,7) -- (-1.5,7) -- (-1.5,1);

    \node[below left] at (-4,0) {$c_0$};
    \draw[fill=black] (-4,0) circle (0.15);

    \node[below right] at (4,0) {$c_1$};
    \draw[fill=black] (4,0) circle (0.15);

    \node[above right] at (4,8) {$r_1$};
    \draw[fill=blue] (4,8) circle (0.15);

    \node[above left] at (-4,8) {$r_0$};
    \draw[fill=blue] (-4,8) circle (0.15);
    
    \draw[postaction={on each segment={mid arrow}}] (-4,0) -- (-2,2) -- (-2,6) -- (-4,8);
    \draw[postaction={on each segment={mid arrow}}] (4,0) -- (2,2) -- (2,6) -- (4,8);
    \draw[postaction={on each segment={mid arrow}}] (-2,2) -- (0,4) -- (2,6);
    \draw[postaction={on each segment={mid arrow}}] (2,2) -- (0,4) -- (-2,6);
    
    \draw[dashed] (2,2) -- (0,0) -- (-2,2);
    \node[below] at (0,0) {$\Psi_{LR}$};
    
    \draw[fill=yellow] (-2,2) circle (0.3);
    \draw[fill=yellow] (2,2) circle (0.3);
    \draw[fill=yellow] (-2,6) circle (0.3);
    \draw[fill=yellow] (2,6) circle (0.3);
    
    \node[above left] at (-3,1) {$A_0$};
    \node[above right] at (3,1) {$A_1$};
    
    \node[below left] at (-3,7) {$B_0$};
    \node[below right] at (3,7) {$B_1$};
    
\end{tikzpicture}
\end{center}
\caption{}
\label{subfig:nonlocalschematic}
\end{subfigure}
\caption{(a) A relativistic quantum task. Time proceeds upwards in the diagram, and the horizontal direction is a spatial dimension. Light rays follow lines with slope $\pm 1$. Input systems $A_0$ and $A_1$ are received at spacetime locations $c_0$ and $c_1$, respectively, and $B_0$ and $B_1$ should be returned at $r_0$ and $r_1$, respectively. The inputs and outputs should be related by some designated channel $\mathcal{N}_{A_0A_1\rightarrow B_0B_1}$. Bob, who issues the challenge, wishes to choose the channel such that Alice is forced to do computations within the gray spacetime region. (b) Completing the task in a local form. The yellow circle represents a channel acting on input systems $A_0$ and $A_1$, and producing output systems $B_0$ and $B_1$. Alice acts within the gray region, corresponding to an honest strategy. (b) A computation happening in the non-local form. $A_0$ is interacted with the $\L$ system, and $A_1$ with the $R$ system, where $\Psi_{LR}$ is entangled. Then, a round of communication is exchanged, and a second round of operations on each side are performed. All operations happen outside of the spacetime region, corresponding to a cheating strategy.}
\label{fig:cylinder1}
\end{center}
\end{figure*}
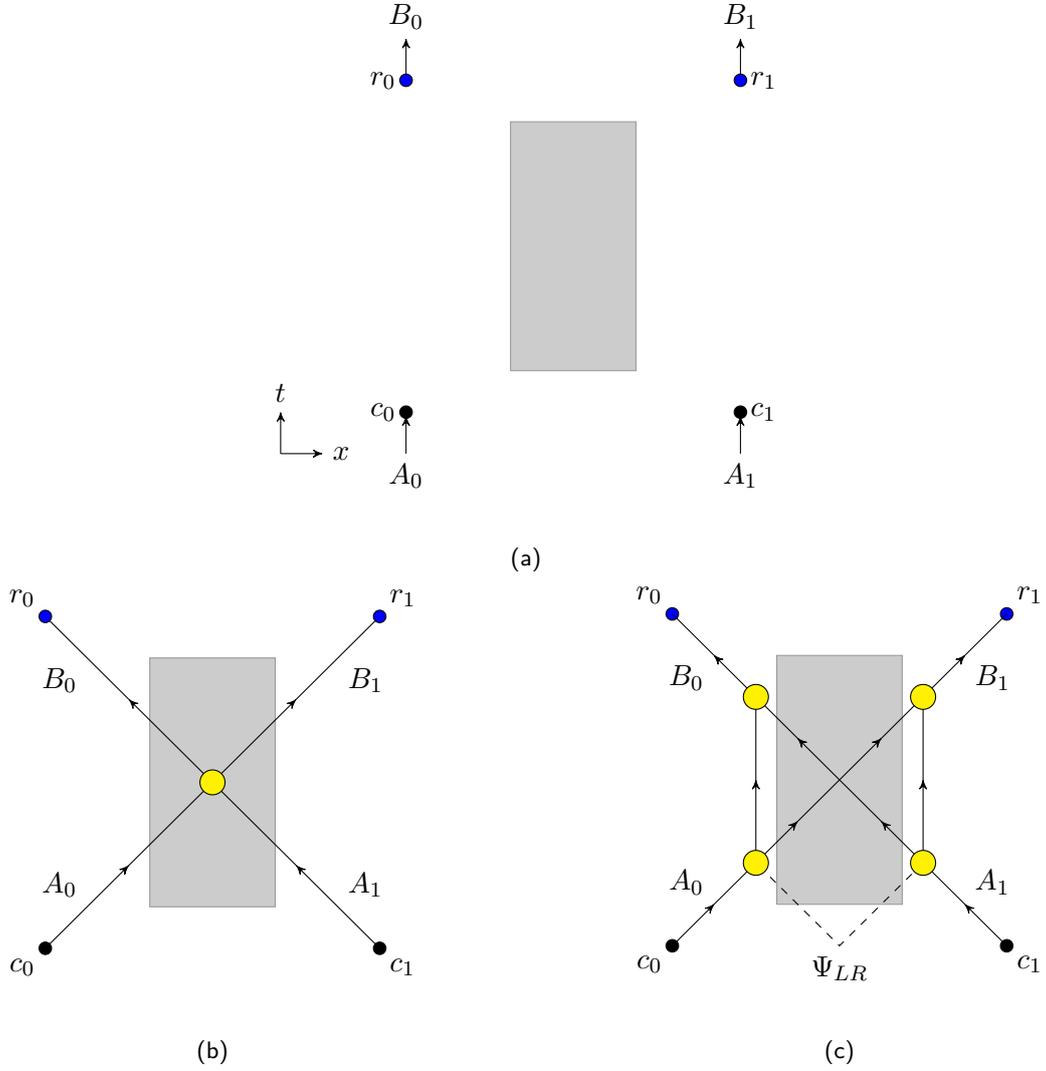

\subsection{Background}

In the cryptographic task of position verification \cite{chandran2009position,kent2011quantum}, a prover (Alice) and verifier (Bob) interact to establish the spatial location of the prover. To do this, Bob issues Alice a challenge, which Bob believes can only be accomplished if Alice applies quantum or classical operations within the spacetime region of interest. The challenge is a relativistic quantum task \cite{kent2012quantum}, with quantum and classical systems input at one set of spacetime locations and another set of input and output systems returned at a second, later set of spacetime points. 

We illustrate the typical position verification set-up in \cref{subfig:PVset-up}. At spacetime locations $c_0$ and $c_1$, which are spatially separated but occur at the same time, inputs $A_0$ and $A_1$ are transmitted by Bob and sent towards the grey shaded region. Then, Alice should process those inputs in some way and return the output systems $B_0$ and $B_1$ to spacetime locations $r_0$ and $r_1$. To complete this, Alice can either act honestly or dishonestly.\footnote{Note that it is more standard to label Bob's role as ``the attacker'', honest Alice's role as ``the prover'', and dishonest Alice's role as ``the cheater''. Our Alice and Bob language is closer to the `quantum tasks' language of \cite{kent2011quantum}, a more general framework within which position-verification can be understood.} If behaving honestly, Alice enters the shaded spacetime region, receives both the inputs and locally acts on them, as shown in \cref{subfig:localschematic}. If behaving dishonestly, Alice sends agents to either side of the grey region, intercepts both transmissions, and then acts in the non-local form shown in \cref{subfig:nonlocalschematic}. This involves local actions on each side of the region, possibly making use of pre-shared randomness or entanglement, and a single, simultaneous round of communication - a computation performed in this form we call a \emph{non-local (quantum) computation}. For a given choice of input state and transformation expected to be performed by Alice, acting in this non-local form may be sufficiently challenging so as to rule out this possibility. If so, then Bob has successfully verified that Alice acts within the specified region. 

Suppose that the input and output systems are all classical. For concreteness, label the input string at $c_0$ by $x$, and the input string at $c_1$ by $y$. Then the outputs at $r_0$ and $r_1$ are some functions $f_0(x,y)$ and $f_1(x,y)$ of the input strings. It is straightforward to see \cite{chandran2009position} that in this fully classical case it is always possible for Alice to cheat by completing the relativistic task in the form shown in \cref{subfig:nonlocalschematic}. To do so, the strategy is to copy the inputs $x$, $y$, then send one copy and keep the other so that $x$ and $y$ are both held at both output locations. Then, $f_0(x,y)$ is computed at $r_0$ and $f_1(x,y)$ at $r_1$, completing the task. 

Unlike classical information, quantum information cannot be copied \cite{wootters1982single}. Inspired by this, \cite{kent2006tagging, malaney2010location} suggested using position verification schemes with quantum input and output systems. It was realized however that even in the quantum case all relativistic quantum tasks can be completed in the non-local, cheating form shown in \cref{subfig:nonlocalschematic}, see \cite{kent2011quantum, buhrman2014position,beigi2011simplified}. This establishes that position verification cannot be made unconditionally secure, at least within the context of quantum mechanics in a fixed spacetime background and without placing assumptions on the entanglement available to an attacker. 

In the absence of unconditional security, we can look for assumptions under which the scheme may be considered secure. For some relativistic quantum tasks, it can be shown that all cheating strategies require large amounts of entanglement. Given this, one can introduce a security model that assumes a bounded amount of entanglement is shared, and then prove security of a position verification scheme by establishing that entanglement in excess of this bound is required to complete a given quantum task. 

Ideally, the relativistic quantum task used in the context of position verification is easy to complete in the honest strategy, and as hard as possible to complete in the dishonest form. One well studied proposal is $f$-routing, which takes the following form. At $c_0$, a quantum system $Q$ of dimension $d$ is given, along with a classical string $x$ of length $n$. At $c_1$, a classical string $y$ of length $n$ is given. As an output, Alice is required to return system $Q$ at $r_{f(x,y)}$, where $f$ is some fixed function mapping strings of length $2n$ to bits. Notice that to complete the $f$-routing task honestly Alice can bring $Q,x$ and $y$ into the spacetime region, compute $f$, then redirect $Q$ based on the outputs. Thus the quantum part of the strategy is almost trivial. 

Recently Bluhm, Christandl, and Speelman, \cite{bluhm2021position} proved the following statement. Pick a random function $f$. Then with high probability, any cheating strategy to complete the corresponding $f$-routing task requires a shared resource system with a dimension that grows with $n$. Thus by increasing $n$, the honest strategy involves a larger classical computation, but the dishonest strategy involves manipulating larger quantum systems. Assuming classical computations are ``easier'' in some appropriate sense than storing quantum systems, we can establish security of the scheme.   

Entanglement cost in the $f$-routing task exhibits an interesting relationship to classical complexity theory. 
One interesting attack on $f$-routing is the ``garden-hose'' protocol \cite{buhrman2013garden,klauck2014new,arunachalam2021communication}.
In that protocol, the number of EPR pairs needed to perform $f$-routing non-locally, call it $GH(f)$, is related to the memory cost of computing $f$ on a Turing machine. 
\begin{align}
    2^{\text{SPACE}_{(2)}(f)} \leq GH(f) \leq 2^{O(\text{SPACE}_{(2)}(f))}
\end{align}
where
\begin{align}
    \spacet (f) = \min_{\substack{M,\alpha,\beta:\\f(x,y) = M(\alpha(x),\beta(y))}}\text{SPACE}(M)\nonumber.
\end{align}
We note here that $\alpha$ and $\beta$ are arbitrary functions; they appear because Alice may locally manipulate her input strings before beginning a protocol. We refer to application of these functions as \emph{pre-processing}.

This connection between the garden-hose model and complexity theory is also constructive: an algorithm for computing $f$ can be turned into a non-local computation using $2^{O(\text{SPACE}_{(2)}(f))}$ entanglement, and a non-local computation in the garden-hose model can be turned into an algorithm for computing $f$, with memory cost given by $\log GH(f)$. 
This connection also suggests proving strong lower bounds on entanglement in $f$-routing should be challenging, as we would obtain lower bounds on space complexity as a consequence. 

The class of functions that can be implemented efficiently using the garden hose protocol is related to $L$, those functions that can be computed in log-space. 
However, the appearance of pre-processing means the efficiently computable functions are instead given by the class $\L_{(2)}$, defined as follows,
\begin{align}
    L_{(2)} \equiv \{f(x,y): f(x,y)=M(\alpha(x),\beta(y)), M\in L \}.
\end{align}
Note that here $L$ denotes the class of functions computable in space logarithmic in $n$, the length of the strings $x$ and $y$ (not the length of $\alpha(x)$ and $\beta(y)$). 
This is the class of functions for which we can complete the $f$-routing task non-locally using polynomial entanglement within the garden-hose protocol. We can analogously define the class $\P_{(2)}$, polynomial time when allowing pre-processing,
\begin{align}
    P_{(2)} \equiv \{f(x,y): f(x,y)=M(\alpha(x),\beta(y)), M\in P \},
\end{align}
where the $P$ inside the definition refers to functions with runtime polynomial in $n$, the length of $x$ and $y$. 
One consequence of the garden-hose protocol's connection to complexity theory is that certain explicit entanglement lower bounds are expected to be hard to prove. 
For example, given a function $f\in P_{(2)}$, if one showed $f$ requires super-polynomial entanglement, then we would learn that $\L_{(2)} \subsetneq P_{(2)}$. 
Since from the definitions above $L=P$ implies $L_{(2)}= P_{(2)}$, we have that $\L_{(2)} \subsetneq P_{(2)}$ implies $L\subsetneq P$. 
Proving that $L\subsetneq P$ however is a longstanding and difficult problem in computer science. 

Recently, a relationship between position-based cryptography and quantum gravity has been highlighted \cite{may2019quantum, may2020holographic}. As we discuss further in \cite{may2022complexity}, in that context there is a tentative expectation coming from the quantum gravity side that entanglement cost in non-local computation should be related to the complexity of the corresponding local computation. From this perspective, the complexity-entanglement relationship exhibited in the garden-hose protocol is especially interesting, and we were motivated to further study $f$-routing and its relationship to complexity due to that connection. 

The possible relationship between complexity and entanglement in non-local computation is also of practical interest in the context of position verification. 
For instance, consider the security setting in which we assume an attacker has bounded entanglement, but do not otherwise restrict their resources. 
In this setting we are interested in functions which require large entanglement to implement non-locally. 
At the same time, the geometry of a position-verification scenario requires the computation be implementable quickly when performed locally.\footnote{This comment is more precise after reading ahead to equation \ref{eq:scatteringregion}: the honest, local computation must be implementable within the region $J_{01\rightarrow 01}$. The time extent of this region is comparable to the spatial size of the region we are trying to localize an honest party to.}
If the function $f$ has exponential complexity, the honest party may not be able to compute it within the needed amount of time. 
Because it uses a randomly chosen (and hence high complexity) function the Bluhm, Christandl, and Speelman result \cite{bluhm2021position} faces this obstruction to realizing a practical and secure position verification setting. 
For this reason, it is important to understand the entanglement cost for implementing low-complexity functions.

\subsection{Summary of results}

In this paper we give a new strategy for completing the $f$-routing task non-locally, which we call ``code-routing''. 
The basic strategy of the protocol is to encode the input system $Q$ into a quantum secret sharing scheme whose access structure is related to the function $f$. 
The shares of the scheme are then routed on simple functions of single input bits. 
Compared to the existing garden-hose protocol, code-routing uses no more entanglement, and probably less. 
To understand why we make use of a connection between the code-routing strategy and complexity theory. 
We also use the code-routing strategy to establish a new relationship between entanglement cost in $f$-routing and the size of quantum secret sharing schemes.
Throughout the work, we work with $p$-dimensional quantum systems, which we call `qupits', with $p$ any prime.\footnote{We can for example choose $p$ based on the function family we wish to perform the $f$-routing task for.}

Calling the minimal entanglement required to $f$-route $E(f)$, we show
\begin{align}
    E(f) \leq O(SP_{p,(2)}(f))
\end{align}
where
\begin{align}
    SP_{p,(2)} (f) = \min_{\substack{M,\alpha,\beta:\\f(x,y) = M(\alpha(x),\beta(y))}}SP_p(M) ,
\end{align}
and $SP_p(M)$ is the minimal size of a span-program over the field $\mathbb{Z}_p$ that computes $M$. The complexity class of functions that can be computed with polynomial-sized span programs is $\modkl$ (see section \ref{sec:lowerbounds} for a definition), so that here the functions for which we can perform $f$-routing using polynomial entanglement is $\modkl_{(2)}$, where again the added subscript accounts for performing local pre-processing of the inputs.

To understand the relationship between entanglement cost in the garden-hose protocol and code-routing, we note first that\footnote{This and other inclusions stated in this paragraph are explained in section \ref{sec:lowerbounds}.} $\L \subseteq \modkl $, and consequently $\text{L}_{(2)} \subseteq \modkl_{(2)}$. Thus, we can perform $f$-routing efficiently for at least those functions that can be efficiently performed in the garden-hose protocol. Further, it is believed that $\L\subsetneq \modkl$. We recall the evidence for this in section \ref{sec:lowerbounds}. Consequently in considering the classes $\L_{(2)}$ and $\modkl_{(2)}$, a strictly larger class of functions can be used to compute the non-local part of $f$. We believe that as a consequence $\L_{(2)}\subsetneq \modkl_{(2)}$. We explain our intuition for this but cannot show it.  

A further consequence of our protocol is a relationship between the size of quantum secret sharing schemes and entanglement requirements in $f$-routing. 
In particular, a quantum secret sharing scheme records a secret, $Q$, into a set of shares $\{v_1,...,v_n\}$ such that some subsets recover $Q$ and others reveal nothing about it. 
The size of a secret sharing scheme is the sum of the log dimension of all the shares. The structure of the scheme is captured by the \emph{indicator function}, which is defined as a map from subsets of shares to bits, and is $0$ when the subset reveals nothing about the secret and $1$ when the subset reveals the secret. Ideally, one constructs a secret sharing scheme with as small of a size as possible for a given indicator function. 

When considering $f$-routing tasks where $f$ can be realized as an indicator function, we build a code-routing scheme that shows the entanglement requirement $E(f)$ is upper bounded by the size of any secret sharing scheme with $f$ as its indicator function. This can also be understood as a constraint on the size of secret sharing schemes. 

It is also interesting to ask if $\modkl_{(2)}$ is the largest class of functions that can be completed using code-routing protocols with polynomial entanglement. Our protocol that achieves this is a special case of the most general possible code-routing construction, in particular it restricts to a class of secret sharing schemes constructed by Smith \cite{smith2000quantum}.
Assuming only those codes are used, and under further constraints on the protocol, we give some partial converse results. For code-routing protocols where Smith codes are used, we can show their complexity is within $P_{(2)}$. When restricting to protocols that concatenate Smith codes to only $O(1)$ depth, we show their complexity is within $\modkl_{(2)}$. For code-routing protocols using arbitrary codes with $O(1)$ shares, we show their complexity is within $L_{(2)}$. Throughout, we have to assume that a certain measure of the size of the protocol is related polynomially to the entanglement used. These results eliminate some directions in which one can try to use a code-routing protocol to perform $f$-routing on functions of larger complexity, and highlight the remaining possibilities. 

\section{\texorpdfstring{$f$}{TEXT}-routing and code-routing protocols}

\subsection{Definition of the \texorpdfstring{$f$}{TEXT}-routing task}

To describe the $f$-routing task, it will be helpful to consider Alice, who carries out the protocol to be an \emph{agency} with several \emph{agents}. Alice's agents co-operate with one another to complete the task. Similarly, Bob is an agency with several agents, who may move through spacetime along different trajectories. For convenience, we will say for example that Bob gives Alice system $A$ at spacetime location $c_0$. Somewhat more precisely, this means that an agent of Bob's, who is located at $c_0$, gives an agent of Alice's the system $A$. 

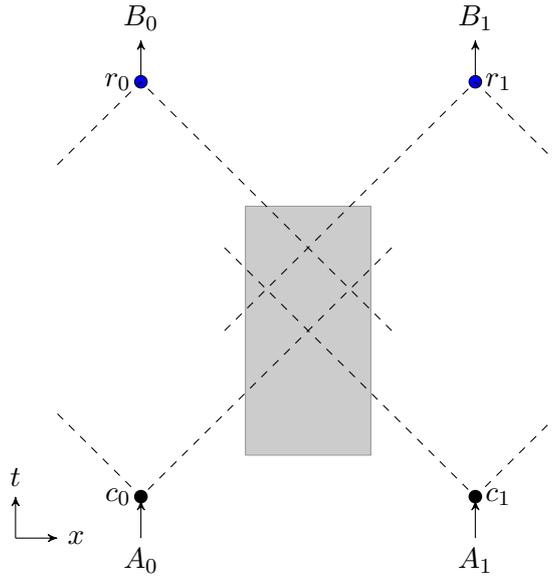
\begin{figure}
    \centering
\begin{tikzpicture}[scale=0.55]
    
    \draw[fill=gray,opacity=0.4] (-1.5,1) -- (1.5,1) -- (1.5,7) -- (-1.5,7) -- (-1.5,1);
    
    \draw[->] (-7,-1) -- (-7,0);
    \node[above] at (-7,0) {$t$};
    \draw[->] (-7,-1) -- (-6,-1);
    \node[right] at (-6,-1) {$x$};
    
    \draw[->] (-4,-1) -- (-4,-0.1);
    \node[below] at (-4,-1) {$A_0$};
    
    \draw[->] (4,-1) -- (4,-0.1);
    \node[below] at (4,-1) {$A_1$};
    
    \node[left] at (-4,0) {$c_0$};
    \draw[fill=black] (-4,0) circle (0.15);
    
    \node[right] at (4,0) {$c_1$};
    \draw[fill=black] (4,0) circle (0.15);
    
    \draw[->] (4,10) -- (4,11);
    \node[above] at (4,11) {$B_1$};
    \draw[->] (-4,10) -- (-4,11);
    \node[above] at (-4,11) {$B_0$};

    \node[right] at (4,10) {$r_1$};
    \draw[fill=blue] (4,10) circle (0.15);

    \node[left] at (-4,10) {$r_0$};
    \draw[fill=blue] (-4,10) circle (0.15);
    
    \node[below] at (0,-0.53) {$ $};

    \draw[dashed] (-2,6) -- (4,0) -- (6,2);
    \draw[dashed] (2,6) -- (-4,0) -- (-6,2);

    \draw[dashed] (6,8) -- (4,10) -- (-2,4);
    \draw[dashed] (-6,8) -- (-4,10) -- (2,4);
    
    \end{tikzpicture}
    \caption{Illustration of the scattering region $J_{01\rightarrow 01}$. The dashed lines extending forward from $c_i$ represent the region $J^+(c_i)$, while the dashed lines extending backwards from the $r_i$ represent the $J^-(r_i)$ regions. Intersecting the four regions $J^+(c_1)$, $J^+(c_2)$, $J^-(r_1)$, $J^-(r_2)$ defines the small diamond sitting within the grey region, which is $J_{01\rightarrow 01}$.}
    \label{fig:scatteringregion}
\end{figure}

The routing task is defined as follows. 
\begin{definition}
An $f$-routing task is defined by a Boolean function $f:\{0,1\}^{2n}\rightarrow \{0,1\}$. The task is carried out by two agencies, Alice and Bob. At spacetime location $c_0$ Bob gives Alice a quantum system $Q$ and a classical string $x$ of length $n$. At spacetime location $c_1$ Bob gives Alice a string $y$. Strings $x$ and $y$ are drawn from the uniform distribution, while $Q$ is in a maximally entangled state $\ket{\Psi^+}_{Q\bar{Q}}$ with reference system $\bar{Q}$ held by Bob. Alice returns a quantum system $B_0$ at location $r_0$ and $B_1$ at $R_1$. Bob measures $\bar{Q}B_{f(x,y)}$ to test if it is in the state $\ket{\Psi^+}$, and Alice completes the task successfully if the test succeeds. 
\end{definition}
When convenient, we will refer to Alice's agent at $c_0$ as Alice$_0$, and Alice's agent at $c_1$ as Alice$_1$. 
As well, it is sometimes convenient to refer to $c_0$ and $r_0$ together collectively as `the left' and $c_1$ and $r_1$ together as the `the right'. 

To complete a routing task, the simplest strategy is to bring $x$, $y$ and $Q$ together, compute the function $f$, and then direct $Q$ based on the result of the computation. To use an $f$-routing task to verify if Alice performs non-trivial operations within a spacetime region $R$, the points $c_0,c_1,r_0,r_1$ should be arranged such that performing this local strategy requires entering $R$. In particular, we define the region
\begin{align}\label{eq:scatteringregion}
    J_{01\rightarrow 01} = J^+(c_0) \cap J^+(c_1) \cap J^-(r_0) \cap J^-(r_1).
\end{align}
Here $J^+(p)$ is the future light cone of $p$, meaning the set of all points $q$ such that information can travel from $p$ to $q$ without moving faster than light, and $J^-(p)$ is the past light cone of $p$, meaning the set of all points $q$ such that one can travel from $q$ to $p$ without travelling faster than the speed of light. 
This is the region in which the input to the local computation of $f$ are available, and the outputs from the computation can still reach the output points. 
Consequently, we choose $c_0,c_1,r_0,r_1$ such that $J_{01\rightarrow 01} \subseteq R$ when we wish to verify Alice can perform computations within $R$. 

\begin{figure*}
    \centering
    \begin{subfigure}{0.45\textwidth}
    \centering
    \begin{tikzpicture}[scale=0.8]
    
    \draw[red, thick] (-1,1)  to [out=-90,in=180] (0,0);
    
    \node[above] at (-1,1) {$Q$};
    \draw[black] plot [mark=*, mark size=2] coordinates{(-1,1)};
    
    \draw[thick] (0,0) -- (5,0);
    \draw[black] plot [mark=*, mark size=2] coordinates{(0,0)};
    \draw[black] plot [mark=*, mark size=2] coordinates{(5,0)};
    
    \draw[thick] (0,-1) -- (5,-1);
    \draw[black] plot [mark=*, mark size=2] coordinates{(0,-1)};
    \draw[black] plot [mark=*, mark size=2] coordinates{(5,-1)};

    \node[right] at (-0.5,0.5) {$x=1$}; 
    
    \draw[blue, thick] (5,0)  to [out=0,in=0] (5.1,-1);
    \node[right] at (5.3,-0.8) {$y=0$};
    
    \end{tikzpicture}
    \caption{}
    \label{fig:disconnectedsurfacesintro}
    \end{subfigure}
    \hfill
    \begin{subfigure}{0.45\textwidth}
    \centering
    \begin{tikzpicture}[scale=0.8]
    
    \draw[blue, thick] (-1,1)  to [out=-90,in=180] (0,0);
    \draw[red, thick] (-1,1)  to [out=-90,in=180] (0,-2);
    
    \node at (-2,-1) {$x=1$};
    
    \node[above] at (-1,1) {$Q$};
    \draw[black] plot [mark=*, mark size=2] coordinates{(-1,1)};
    
    \draw[thick] (0,0) -- (5,0);
    \draw[black] plot [mark=*, mark size=2] coordinates{(0,0)};
    \draw[black] plot [mark=*, mark size=2] coordinates{(5,0)};
    
    \draw[thick] (0,-1) -- (5,-1);
    \draw[black] plot [mark=*, mark size=2] coordinates{(0,-1)};
    \draw[black] plot [mark=*, mark size=2] coordinates{(5,-1)};

    \node[right] at (-0.5,0.5) {$x=0$}; 
    
    \draw[blue, thick] (5,0)  to [out=0,in=0] (5.1,-1);
    \node[right] at (5.3,-0.8) {$y=0$};
    
    \draw[thick] (0,-2) -- (5,-2);
    \draw[black] plot [mark=*, mark size=2] coordinates{(0,-2)};
    \draw[black] plot [mark=*, mark size=2] coordinates{(5,-2)};
    
    \end{tikzpicture}
    \caption{}
    \label{fig:connectedsurfacesintro}
    \end{subfigure}
    \caption{Some simple garden-hose protocols. Blue lines indicate Bell basis measurements. Black lines indicate shared EPR pairs, with the left side of the pairs held by Alice$_0$ and right side held by Alice$_1$. a) Garden-hose protocol for computing $AND(x,y)$. Alice$_0$ measures $Q$ and the first EPR pair in the Bell basis iff $x=1$. Alice$_1$ measures the two EPR pairs iff $y=0$. b) Garden-hose protocol for $OR(x,y)$, which uses similar conditional measurements.}
    \label{fig:GHexamples}
\end{figure*}
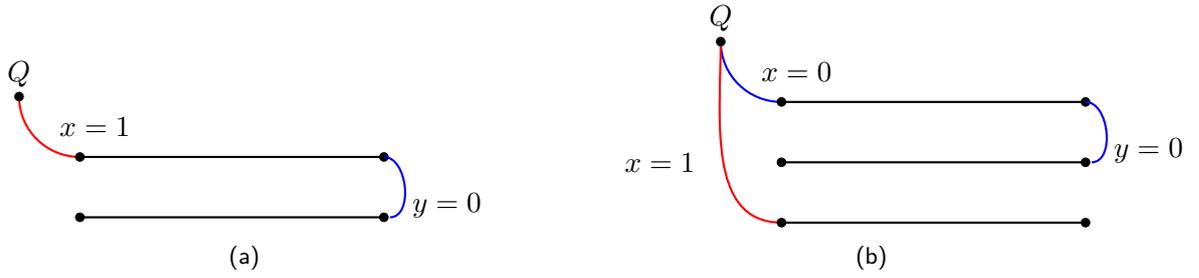

To perform the routing task non-locally, the best known strategy is the garden-hose protocol \cite{buhrman2013garden}.
It involves sharing EPR pairs between $c_0$ and $c_1$, then doing a set of Bell measurements on pairs of entangled particles.
Which measurements are performed depends on the values of the strings $x$ and $y$.
The measurement outcomes are then communicated to both of the output locations.
If the mappings from strings $x$, $y$ to a set of measurements on both sides is chosen correctly, it will be possible to recover the system $Q$ at $r_{f(x,y)}$.
We give simple examples of computing a NOT and AND function in \cref{fig:GHexamples}. As discussed in the introduction, the entanglement cost of completing the $f$-routing task using the garden-hose protocol is controlled by the space complexity of $f$. 

Another possible attack is given in \cite{beigi2011simplified}. 
This attack also works for arbitrary quantum tasks. Applied to $f$-routing, it has exponential in $n$ entanglement cost for any choice of function $f$. 

\subsection{Code-routing protocols}

Error-correcting codes are a standard tool appearing throughout quantum information theory --- here we consider their use in performing the $f$-routing task.
Because only two parties (an agent on the left and on the right) are involved in a non-local computation, it is unclear why error-correcting codes should be related to non-local computation. 
However, we are motivated to do this because of a recent connection \cite{may2019quantum} between non-local quantum computation and the AdS/CFT correspondence \cite{maldacena1999large,witten1998anti}.
Error-correction plays an important role in the AdS/CFT correspondence, suggesting a connection between non-local computation and error-correction.
We study a family of $f$-routing protocols that exploit error correction, which we call code-routing protocols.
After giving the general form of any such protocol, we discuss a particular class of codes that expands the set of computations performable using polynomial entanglement to $\modkl_{(2)}$, a complexity class which is known to be at least as large as $\text{L}_{(2)}$, and is probably larger.

The basic structure of a code-routing protocol involves recording $Q$ into an error-correcting code, then sending the shares of that code to the left or right based on the input variables. 
We can also carry out garden-hose type strategies on individual shares, or record those shares into subsequent codes, including choosing which encoding to use based on the input variables. 

The simplest example of a code-routing protocol, which we will use as a subroutine in subsequent constructions, is `unit-routing'. The functionality of the unit-routing protocol is to send a share $v_i$ to the side labelled by a bit $z_j$. We explain how to perform the unit-routing protocol in \cref{fig:unitrouting}.

We describe the most general form of a code-routing protocol below.\footnote{The reader may wish to skip this detailed definition and return to it after understanding some of the simple examples below.}
\begin{definition}\label{def:crprotocol}\textbf{Code-routing protocol:}
A code-routing protocol is defined by two maps $C_0[x]$ and $C_1[y]$, each mapping from input strings of length $n$ to a tuple,
\begin{align}
    C_0[x] &: \{0,1\}^n \rightarrow (a(x),S_0,...,S_{\ell}), \nonumber \\
    C_1[y] &: \{0,1\}^n \rightarrow (b(y),S_{\ell+1},...,S_{\ell+\ell'}). \nonumber 
\end{align}
The combined outputs 
\begin{align}
    I=(a(x),S_0,...S_\ell,b(y),S_{\ell+1},...S_{\ell+\ell'})
\end{align}
we refer to as the \emph{protocol tape}. 
Each $S_i$ corresponds to one encoding, teleportation, or `unit-routing' of a local share.
We denote it as a tuple $S_i = (v_i,\{w_i^j\},T_i)$, with $v_i$ a label for an input share, $\{w_i^j\}$ a set of output shares, and $T_i$ a description of an encoding, teleportation, or `unit-routing'. 
Define $n_i=|\{w_i^j\}|$ to be the number of output shares associated with $S_i$. Then:
\begin{itemize}
    \item When $n_i=0$, $T_i$ describes a unit-routing or keep/send instruction. For a unit-routing $T_i$, will be the label of a single bit of $a(x)$ or $b(y)$, or its negation. For a keep/send instruction, $T_i$ will be a $0$ or $1$ indicating that the share should be brought to $r_0$ or $r_1$. 
    \item When $n_i=1$, $T_i$ will be empty, and the tuple $(v_i,w_i^0,\emptyset)$ describes a teleportation from the $v_i$ system onto the $w_{i}^0$ system.\footnote{Note that here, the Pauli correction which is required in the teleportation protocol is implemented in the final stage of the non-local computation. Subsequent encodings, teleportations, or unit-routings of the share $w_i^0$ will take place \emph{before} this correction is performed.}
    \item When $n_i>1$, $T_i$ describes an encoding into an error-correcting code, with the $w_i^j$ systems the output systems of the encoding procedure.
\end{itemize}
Alice$_0$ and Alice$_1$ carry out the code-routing protocol by computing $C_0[x]$ and $C_1[y]$, then encoding, teleporting, or unit-routing each share according to the pattern described by the protocol tape. 
\end{definition}
Code-routing includes the garden-hose protocol as a special case: if no systems are put into codes, the remaining protocol amounts to a set of choices about which pairs of entangled systems should be measured in the Bell basis, as in the garden-hose protocol.
This shows code-routing uses at most as much entanglement as the garden-hose.
More generally, including non-trivial encodings allows a larger class of strategies. 

\begin{figure*}
    \centering
    \begin{subfigure}{0.45\textwidth}
    \centering
    \begin{tikzpicture}[scale=0.4]
    
    \draw[thick] (0,0) -- (0,3) -- (3,3) -- (3,0) -- (0,0);
    \node at (1.5,1.5) {\large{$x$}};
    
    \draw[thick] (1.5,3) -- (1.5,4.5);
    \draw[black] plot [mark=*, mark size=3] coordinates{(1.5,4.5)};
    
    \draw[gray,<->] (4.5,1.5) -- (7.5,1.5);
    
    \draw[black] plot [mark=*, mark size=3] coordinates{(9,1.5)};
    \draw[blue,->] (9,1.5) -> (11,3.5);
    \node[above right, align=center] at (11,3.5) {if $x=1$, \\ send};
    
    \node[above] at (1.5,4.5) {$v$};
    \node[below] at (9,1.5) {$v$};
    
    \end{tikzpicture}
    \caption{}
    \label{fig:unitroutinga}
    \end{subfigure}
    \hfill
    \begin{subfigure}{0.45\textwidth}
    \centering
    \begin{tikzpicture}[scale=0.4]
    
    \draw[thick] (0,0) -- (0,3) -- (3,3) -- (3,0) -- (0,0);
    \node at (1.5,1.5) {\large{$y$}};
    
    \draw[thick] (1.5,3) -- (1.5,4.5);
    \draw[black] plot [mark=*, mark size=3] coordinates{(1.5,4.5)};
    
    \draw[gray,<->] (4.5,1.5) -- (7.5,1.5);
    
    \draw[black] plot [mark=*, mark size=3] coordinates{(9,4.5)};
    \draw[thick] (9,4.5) to [out=-90,in=180] (11.9,1.5);
    \draw[black] plot [mark=*, mark size=3] coordinates{(12,1.5)};
    \draw[black,thick] (12,1.5) -- (16,1.5);
    \draw[black] plot [mark=*, mark size=3] coordinates{(16,1.5)};
    
    \draw[blue,->] (16,1.5) -- (14,3.5);
    \node[above right, align=center] at (14,3.5) {if $y=0$,\\ send};
    
    \node[above] at (1.5,4.5) {$v$};
    \node[above] at (9,4.5) {$v$};
    
    \end{tikzpicture}
    \caption{}
    \label{fig:unitroutingb}
    \end{subfigure}
    \caption{Illustration of the \emph{unit-routing} protocol. The effect of the protocol is to bring the share $v$ to the side labelled by the input bit. a) For an input bit $z=x_i$ held by Alice$_0$, who holds share $v$, the share is sent to $c_1$ during the communication round iff $z=1$. b) With $v$ at $c_0$ but input bit $z=y_j$ held at $c_1$, the share is first measured in the Bell basis with one end of an EPR pair that has been shared between $c_0$ and $c_1$. After the measurement, the systems at $c_1$ holds the information on $v$ up to a Pauli correction. Call this system $v'$. During the communication round, Alice$_1$ sends $v'$ to $c_0$ if $z=0$, and keeps it if $z=1$. Simultaneously, Alice$_0$ maintains a copy of her measurement outcome and sends a copy to the right. On whichever side $v'$ has been brought to, the local agent can undo the Pauli correction and recover $v$. Notice that the bit $z$ could also be the NOT of one of the input bits received by Alice$_0$ or Alice$_1$. Similar protocols are used when $v$ is held by Alice$_1$.}
    \label{fig:unitrouting}
\end{figure*}
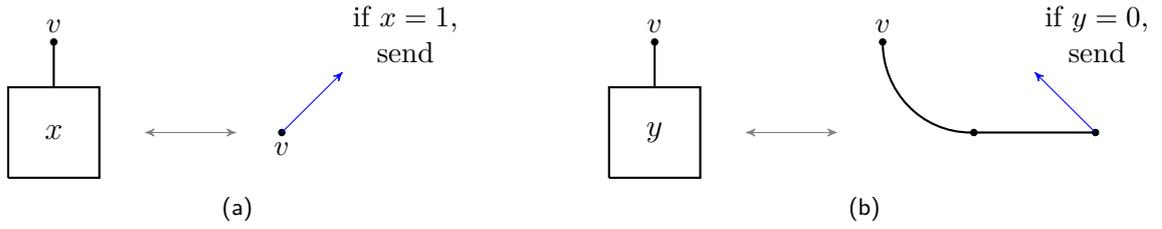

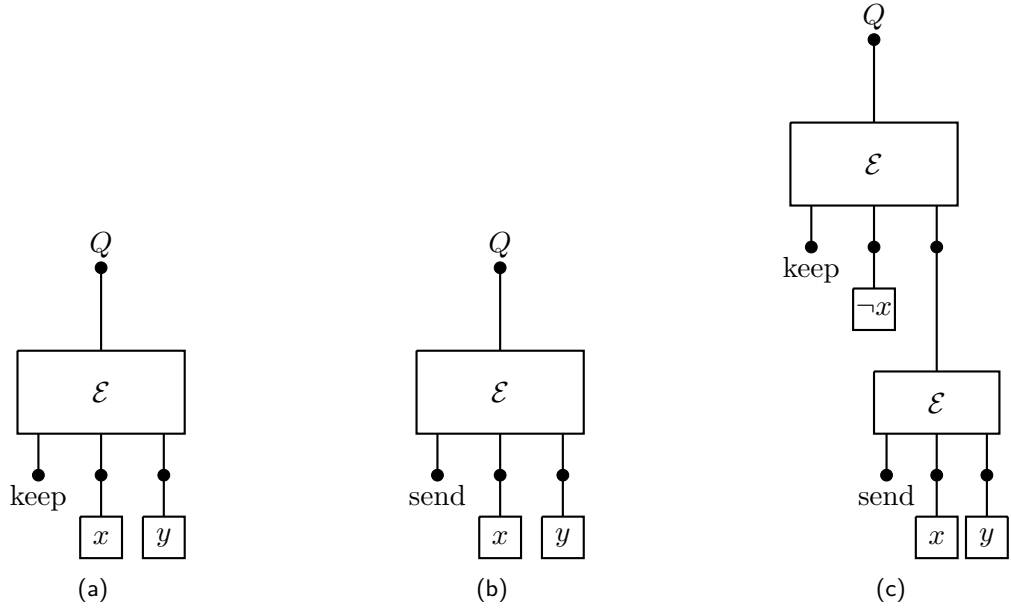
\begin{figure*}
    \centering
    \begin{subfigure}{0.3\textwidth}
    \centering
    \begin{tikzpicture}[scale=1.1]
    
    \draw[black,thick] (-3,3) -- (-1,3) -- (-1,2) -- (-3,2) -- (-3,3);
    \node at (-2,2.5) {$\mathcal{E}$};
    \draw[black,thick] (-2,4) -- (-2,3);
    \draw[black] plot [mark=*, mark size=2] coordinates{(-2,4)};
    
    \draw[black,thick] (-1.25,2) -- (-1.25,1.5);
    \draw[black,thick] (-2,2) -- (-2,1.5);
    \draw[black,thick] (-2.75,2) -- (-2.75,1.5);
    
    \draw[black] plot [mark=*, mark size=2] coordinates{(-1.25,1.5)};
    \draw[black] plot [mark=*, mark size=2]
    coordinates{(-2,1.5)};
    \draw[black] plot [mark=*, mark size=2] coordinates{(-2.75,1.5)};
    \node[below] at (-2.75,1.5) {keep};
    
    \node[above] at (-2,4) {$Q$};
    
    \draw[thick] (-1.5,1) -- (-1,1) -- (-1,0.5) -- (-1.5,0.5) -- (-1.5,1);
    \draw[thick] (-1.25,1.5) -- (-1.25,1);
    \node at (-1.25,0.75) {$y$};
    
    \draw[thick] (-2.25,1) -- (-1.75,1) -- (-1.75,0.5) -- (-2.25,0.5) -- (-2.25,1);
    \draw[thick] (-2,1.5) -- (-2,1);
    \node at (-2,0.75) {$x$};
    
    \end{tikzpicture}
    \caption{}
    \label{fig:AND}
    \end{subfigure}
    \hfill
    \begin{subfigure}{0.3\textwidth}
    \centering
    \begin{tikzpicture}[scale=1.1]
    
    \draw[black,thick] (-3,3) -- (-1,3) -- (-1,2) -- (-3,2) -- (-3,3);
    \node at (-2,2.5) {$\mathcal{E}$};
    \draw[black,thick] (-2,4) -- (-2,3);
    \draw[black] plot [mark=*, mark size=2] coordinates{(-2,4)};
    \node[above] at (-2,4) {$Q$};
    
    \draw[black,thick] (-1.25,2) -- (-1.25,1.5);
    \draw[black,thick] (-2,2) -- (-2,1.5);
    \draw[black,thick] (-2.75,2) -- (-2.75,1.5);
    
    \draw[black] plot [mark=*, mark size=2] coordinates{(-1.25,1.5)};
    \draw[black] plot [mark=*, mark size=2]
    coordinates{(-2,1.5)};
    \draw[black] plot [mark=*, mark size=2] coordinates{(-2.75,1.5)};
    
    \node[below] at (-2.75,1.5) {send};
    
    \draw[thick] (-1.5,1) -- (-1,1) -- (-1,0.5) -- (-1.5,0.5) -- (-1.5,1);
    \draw[thick] (-1.25,1.5) -- (-1.25,1);
    \node at (-1.25,0.75) {$y$};
    
    \draw[thick] (-2.25,1) -- (-1.75,1) -- (-1.75,0.5) -- (-2.25,0.5) -- (-2.25,1);
    \draw[thick] (-2,1.5) -- (-2,1);
    \node at (-2,0.75) {$x$};
    
    \end{tikzpicture}
    \caption{}
    \end{subfigure}
    \hfill
    \begin{subfigure}{0.3\textwidth}
    \centering
    \begin{tikzpicture}[scale=1.1]
    
    \draw[black,thick] (-3,3) -- (-1,3) -- (-1,2) -- (-3,2) -- (-3,3);
    \node at (-2,2.5) {$\mathcal{E}$};
    \draw[black,thick] (-2,4) -- (-2,3);
    \draw[black] plot [mark=*, mark size=2] coordinates{(-2,4)};
    \node[above] at (-2,4) {$Q$};
    
    \draw[black,thick] (-1.25,2) -- (-1.25,1.5);
    \draw[black,thick] (-2,2) -- (-2,1.5);
    \draw[black,thick] (-2.75,2) -- (-2.75,1.5);
    
    \draw[black] plot [mark=*, mark size=2] coordinates{(-1.25,1.5)};
    \draw[black] plot [mark=*, mark size=2]
    coordinates{(-2,1.5)};
    \draw[black] plot [mark=*, mark size=2] coordinates{(-2.75,1.5)};
    
    \node[below] at (-2.75,1.5) {keep};
    
    \draw[thick] (-1.25,1.5) -- (-1.25,0);
    
    \draw[thick] (-2.25,1) -- (-1.75,1) -- (-1.75,0.5) -- (-2.25,0.5) -- (-2.25,1);
    \draw[thick] (-2,1.5) -- (-2,1);
    \node at (-2,0.75) {$\neg x$};
    
    \draw[thick] (-2,0) -- (-0.5,0) -- (-0.5,-0.75) -- (-2,-0.75) -- (-2,0);
    \node at (-1.25,-0.375) {$\mathcal{E}$};
    
    \draw[thick] (-1.25,-0.75) -- (-1.25,-1.25);
    \draw[black] plot [mark=*, mark size=2] coordinates{(-1.25,-1.25)};
    
    \draw[thick] (-0.65,-0.75) -- (-0.65,-1.25);
    \draw[black] plot [mark=*, mark size=2] coordinates{(-0.65,-1.25)};
    
    \draw[thick] (-1.85,-0.75) -- (-1.85,-1.25);
    \draw[black] plot [mark=*, mark size=2] coordinates{(-1.85,-1.25)};
    \node[below] at (-1.85,-1.25) {send};
    
    \draw[thick] (-1.25,-1.25) -- (-1.25,-1.75);
    \draw[thick] (-0.65,-1.25) -- (-0.65,-1.75);
    
    \draw[thick] (-0.4,-1.75) -- (-0.9,-1.75) -- (-0.9,-2.25) -- (-0.4,-2.25) -- (-0.4,-1.75);
    \node at (-0.65,-2) {$y$};
    
    \draw[thick] (-1.5,-1.75) -- (-1,-1.75) -- (-1,-2.25) -- (-1.5,-2.25) -- (-1.5,-1.75);
    \node at (-1.25,-2) {$x$};
    
    \end{tikzpicture}
    \caption{}
    \label{fig:CRformula}
    \end{subfigure}
    \caption{Some simple code-routing protocols. The map $\mathcal{E}$ takes in the $Q$ system and records it into a 3 share secret sharing scheme where any 2 shares recover the secret. In the protocol, Alice$_0$, who initially holds $Q$, performs the encoding map $\mathcal{E}$. The lower boxes indicate the unit-routing protocol should be implemented on the attached shares. a) Code-routing protocol for computing $AND(x,y)$. The protocol tape describing this protocol consists of the tuples $S_1=(Q,\{A,B,C\},\mathcal{E})$, $S_2=(A,\{\},0)$, $S_3=(B,\{\},x)$, $S_4=(C,\{\},y)$. b) Code-routing protocol for $OR(x,y)$. c) Code-routing protocol for computing $f(x,y)=AND(NOT(x),OR(x,y))$. This method of concatenating codes to generalizes to arbitrary Boolean formulas. The entanglement cost is bounded above by the formula size.}
    \label{fig:CRexamples}
\end{figure*}

To understand code-routing, it will be helpful to begin with simple examples and build up to more elaborate constructions. 
Some basic examples of code-routing protocols are shown in \cref{fig:CRexamples}.
There, we $f$-route on the AND and OR functions using an erasure code on $3$ shares that corrects one erasure error.
The protocols for AND and OR given here can be compared to the garden-hose strategies for computing the same functions in \cref{fig:GHexamples}.

One convenient property of the code-routing strategy is that composition of functions is implemented in a simple way.
To see this, consider a simple example, which is easy to generalize.
Consider the function $f(x,y)=AND(NOT(x),OR(x,y))$.
To execute this in a code-routing protocol, one can use the code shown in \cref{fig:CRformula}.
Notice that we concatenate codes according to the pattern given by the Boolean formula for function $f(x,y)$.
This generalizes to any Boolean formula, although we must use DeMorgans' laws to move the NOT gates to the input layer.
This shows that the entanglement cost for code-routing on a function $f$ is bounded above by the formula size of $f$, where by formula size we mean the number of inputs to the formula, counted with repetition.\footnote{E.g. $f(x,y)=AND(NOT(x),OR(x,y))$ has size $3$.}

Building on the AND and OR examples, we can replace the simple threshold code with other, more structured examples. An interesting class of examples is constructed from quantum secret sharing schemes, which we review briefly before describing the protocol. 

A quantum secret sharing scheme is a quantum error-correcting code with the additional feature that collections of subsystems are either authorized, meaning they can be used to recover the encoded state, or unauthorized, meaning they reveal \emph{no} information about the state.
The set of authorized sets for a given secret sharing scheme is known as its \emph{access structure}.
Call the shares produced by the secret sharing scheme $\{v_i\}_i$.
Then the scheme's access structure defines a corresponding indicator function $f_I$ according to
\begin{align}
f_I(z)=
\begin{cases}
    0, & \bigcup_{z_i : z_i=1} v_{i}\,\,\, \text{is unauthorized} \\
    1, & \bigcup_{z_i : z_i=1} v_{i} \,\,\, \text{is authorized} \,
\end{cases}
\end{align}
All valid indicator functions satisfy two constraints.
First, the no-cloning theorem implies no two disjoint subsets can recover the state.
At the level of the indicator function, this is expressed as $f(z)=1\implies f(\bar{z})=0$.
Second, adding additional shares to a set never prevents recovery, which implies $f(z)$ is monotone.\footnote{A Boolean function $f(z)$ is said to be monotone if $x\preceq y \implies f(x)\leq f(y)$, where $x\preceq y$  means that $x_i\leq y_i$ for all $i$.} 
In \cite{gottesman2000theory}, it was shown that whenever the indicator function is no-cloning and monotone, it is possible to construct a corresponding quantum secret sharing scheme. 
Finally, define the \emph{size} of a quantum secret sharing scheme to be the sum of the log dimension of all the shares. For shares built from qubits, this is the total number of qubits the secret is encoded into. 

Using a code-routing protocol based on a single encoding of $Q$ into a quantum secret sharing scheme, we can prove the following theorem. 
\begin{theorem}
Consider an $f$-routing task where $f$ is a valid indicator function. Then the entanglement cost of completing the routing task for $f$ is upper bounded by the size of any quantum secret sharing scheme that has $f$ as its indicator function. 
\end{theorem}
\begin{proof}\,
Construct an $f$-routing protocol as follows. On the left, record $Q$ into a quantum secret sharing scheme with shares $\{v_1,...,v_{2n}\}$ and indicator function $f(x,y)$. In particular, use the isometric extension of the encoding map, and have Alice$_0$ hold the purifying system $R$. Then, for $1\leq i\leq n$ carry out the unit-routing protocol on each share $v_i$ with $x_i$ as input. For $n+1\leq i \leq 2n$ unit-route share $v_i$ on $y_i$. We will show this procedure correctly completes the $f$-routing task, and has entanglement cost upper bounded by the size of any secret sharing scheme with indicator $f$. 

For correctness, notice that by construction Alice$_1$ obtains the set of shares $K(x,y)\equiv \bigcup_{z_i : z_i=1} v_{i}$, and Alice$_0$ holds the purification of $K(x,y)$ (consisting of the remaining shares plus $R$). If $f(x,y)=1$, by construction we have that $K(x,y)$ is authorized, so Alice$_1$ recovers $Q$, which is correct. If $f(x,y)=0$ Alice$_1$ receives an unauthorized set of shares. This ensures all systems held by Alice$_1$ reveal nothing about $Q$. Since Alice$_0$ holds the purifying system, by decoupling \cite{schumacher1996quantum,schumacher2002approximate} we have that Alice$_0$ can recover $Q$. Again this is correct.

To understand the entanglement cost of this protocol, notice that the unit-routing of share $v_i$ for $i > n$ requires $\log_2 d_{v_i}$ EPR pairs. Unit-routing on shares $v_i$ for $i\leq n$ has no entanglement cost, since the needed bits $x_i$ are held locally. The total entanglement cost is just the entanglement cost of all the unit-routings, giving
\begin{align}
    E(f) \leq \sum_{n < i\leq 2n} \log_2 d_{v_i} \leq \sum_{ 1\leq i\leq 2n} \log_2 d_{v_i}.
\end{align}
The right hand side is just the size of the secret sharing scheme used, so we are done. 
\end{proof}

For a given indicator function, the most efficient quantum secret sharing scheme is the one due to Smith \cite{smith2000quantum}. In particular, Smith's scheme has size $O(mSP_p(f))$, where $mSP_p(f)$ is the size of a \emph{monotone span program} over $\mathbb{Z}_p$ that computes $f$. We define span programs in \cref{appendix:spanprograms}. This shows that for indicator functions the entanglement cost of $f$-routing is upper bounded by monotone span program size.

Next, we continue to progress towards more elaborate code-routing protocols, which will allow us to do code-routing for arbitrary functions, not just indicator functions. 
In particular, we will introduce unit-routings that direct a share based on the negation of one of the input bits, rather than an input bit directly, which will allow us to route on non-monotone functions. 
As well, we will route on functions which violate the no-cloning property by realizing them as restrictions of functions which do have the no-cloning property. 
Combining these tools we prove the following theorem. 

\begin{theorem}\label{thm:codeprotocol}
Using a code-routing protocol, the routing task can be completed for any function $f$ using a resource state consisting of $O(SP_{p,(2)}(f))$ maximally entangled qupits, where
\begin{align}
    SP_{p,(2)}(f)= \min_{h,\alpha,\beta} \{ SP_p(h):f(x,y)=h(\alpha(x),\beta(y))\},\nonumber 
\end{align}
and $SP_{p}(h)$ is the size the smallest span program over the field $\mathbb{Z}_p$ computing $h$.
\end{theorem}
In the next section we show that span program size is no larger than the entanglement cost in the garden-hose protocol, and given some complexity theoretic assumptions is smaller. 

Towards proving this theorem, we build a routing protocol in the following way.
We show that $f$ can be expressed as $f(z)=(f_I\circ g)(z,b)$, with $b$ a single bit, $z=(x,y)$, $g$ maps $(z,b)$ to $(z,\neg z, b)$, and $f_I$ an indicator function. We state this in the next lemma. 
\begin{lemma}
\label{lemma:decomposition}
\,Given a function $f:\{0,1\}^{m}\rightarrow \{0,1\}$, there exist functions 
\begin{align}
    f' &: \{0,1\}^{m+1} \rightarrow \{0,1\}, \nonumber \\
    f_I &:\{0,1\}^{2m+1}\rightarrow \{0,1\}, \nonumber \\
    g &:\{0,1\}^{m+1} \rightarrow \{0,1\}^{2m+1}, \nonumber 
\end{align}  
such that 
\begin{itemize*}
    \item $f'(z,1) = f(z)$
    \item $f'(z,b) = f_I\circ g(z,b)$
    \item $f_I$ is a valid indicator function
    \item $g$ acts on the first $m$ bits of its input by copying each bit $z_i$ and negating one copy, $z_i\rightarrow (z_i,\neg z_i)$. It leaves the final bit $b$ unchanged.
    \item $mSP_p(f_I)\leq SP_p(f)+1$, where $SP_p(h)$ denotes the minimal size of a span program over $\mathbb{Z}_p$ computing $h$, and $mSP_p(h)$ the size of a monotone span program computing $h$.
\end{itemize*}
\end{lemma}
\noindent We prove this lemma in \cref{appendix:decompositionlemma}. 

Using this lemma, we are ready to prove \cref{thm:codeprotocol}. 

\vspace{0.5cm}
\begin{proof}\, \textbf{(Of \cref{thm:codeprotocol})}
Let the function we will perform the routing task on be $f(x,y)$.
We can first allow Alice$_0$ and Alice$_1$ to apply local functions to their strings $x$ and $y$, producing new strings $\alpha(x)$ and $\beta(y)$.
These are chosen, along with a function $h$, such that $f(x,y)=h(\alpha(x),\beta(y))$.
Let $m=|\alpha|+|\beta|$.

\Cref{lemma:decomposition} gives that we can realize $h$ as a restriction of $h'(z,b)=h_I\circ g(z,b)$, with $h_I$ a valid indicator function, and $g(z)$ mapping $m+1$ bits to $2m+1$ bits.
For the indicator function $h_I$, use the construction in Ref.~\cite{smith2000quantum} to find an encoding map $\mathcal{E}_{Q\rightarrow V}$ which prepares a secret sharing scheme with access structure corresponding to $h_I$.

The protocol is as follows.
After receiving $Q$, Alice$_0$ applies the isometric extension of the encoding channel, call it $V_{Q\rightarrow VE}^{\mathcal{E}}$.
This produces output systems $v_i$, $1 \leq i \leq 2m+1$, and $E$.
The environment system $E$ is retained by Alice$_0$.
Then, Alice$_0$ and Alice$_1$ carry out the unit-routing protocol (see \cref{fig:unitrouting}) to bring share $v_i$ to Alice$_{g(z)_i}$, where by $g(z)_i$ we mean the $i$th bit of $g(z)$.\footnote{Recall that the $i$th bit of $g(z)$ is either a bit of the input $z$ or a negated copy, such that $z_j = g(z)_{2j}$ if $i$ is even or $\neg z_j = g(z)_{2j+1}$ if $i$ is odd.}
Note that we always take $z_{2m+1}=b=1$ and $g$ to always act trivially on this bit, so that share $v_{2m+1}$ is always sent to Alice$_1$.

Next we verify that this protocol works correctly, in that $Q$ will be recovered on Alice$_{f(z)}$'s side.
Consider that Alice$_1$ holds all those shares $v_i$ such that $g(z_i)_i=1$.
If this is an authorized set, she will be able to recover $Q$.
By design, this occurs exactly when $h'(z,1)=h_I(g(z,1))=1$, and by construction $h'(z,1)=h(z)$, so this is correct.
Alternatively if the set of shares $v_i$ such that $g(z_i)=1$ is unauthorized, then Alice$_1$'s systems reveal nothing about the encoded state.
Because Alice$_0$ performed the encoding procedure isometrically and retained the environment, decoupling ensures that Alice$_0$ can now recover the state.
This occurs exactly when $h'(z,1)=h_I(g(z,1))=0$, so $h(z)=0$, and again this is correct.

Finally we determine the entanglement cost of performing this protocol.
All the entanglement use occurs in teleporting shares $v_i$, $2|\alpha| < i \leq  2 m$ from Alice$_0$ to Alice$_1$, which occurs as part of the unit-routing protocol.
The required entanglement depends on the size of the shares $v_i$, which in turn depends on the details of the secret-sharing scheme construction.
Specifically, the protocol can be performed using not more than
\begin{align}
    \sum_{2|\alpha| < i\leq 2 m} \log_k d_{v_i} \leq \sum_{1 < i \leq 2 m} \log_k d_{v_i}
\end{align}
maximally entangled pairs of qupits.
For the construction of Ref.~\cite{smith2000quantum}, this is at most $(2 mSP_p(h_I)+1)$.
From \cref{lemma:decomposition} we have also that $mSP_p(h_I)\leq SP_p(h)+1$, completing the proof. 
\end{proof}

\section{Entanglement and complexity in code-routing}

\subsection{Lower bounds on efficiently achievable complexity}\label{sec:lowerbounds}

In the last section we saw that the code-based protocol can carry out a routing task using at most $O(SP_{p,(2)}(f))$ maximally entangled pairs of qupits, where $SP_{p,(2)}(f)$ is the minimal size of a span program over $\mathbb{Z}_p$ (with $p$ prime) that computes the non-local part of $f$.
To capture the set of functions that can be performed using reasonable amounts of entanglement with this strategy, we define the following complexity classes. 
\begin{definition}
For prime $p$, $\PSP_p$ is the set of families of functions $f_n:\{0,1\}^n\rightarrow \{0,1\}$ that can be computed using span programs over the field $\mathbb{Z}_p$ of size polynomial in $n$. 
\end{definition}
\begin{definition}
$\PSP_{p,(2)}$ is the set of families of functions $f_n:\{0,1\}^{2n}\rightarrow \{0,1\}$ which can be computed in the form $f_n(x,y)=h_n(\alpha(x),\beta(y))$ with $h_n\in \PSP_p$. 
\end{definition}

\Cref{thm:codeprotocol} establishes that the routing task can be completed with polynomial EPR pairs for a function family $\{f_n\}$ at least when it is in the class $\PSP_{p,(2)}$, for any prime $p$.
This gives that the class of functions efficiently implementable in the code-routing strategy is at least $\cup_{\text{prime}\,\, p} \PSP_{p,(2)}$.
We are interested in the relationship between this class and $\text{L}_{(2)}$, which is the class of functions that can be computed non-locally in the garden-hose model (the most efficient previously known protocol) with polynomial entanglement.
In the next two sections we give evidence that $\text{L}_{(2)} \subsetneq \PSP_{p,(2)}$, so that code-routing improves on the garden-hose model.\footnote{In the introduction we make the statement that code-routing achieves the class $\modkl_{(2)}$, and that $L \subseteq \modkl$. In fact $PSP_p=\modkl$ as we discuss in this section, so this is the same statement as is made here.}

\vspace{0.3cm}
\noindent \textbf{L and $\PSP_p$}
\vspace{0.3cm}

We will start by considering the classes without local pre-processing of the inputs, $\text{L}$ and $\PSP_p$.
 It is believed that $\text{L} \subsetneq \PSP_p$.
 To understand why, we first need to introduce a few related complexity classes, $\NL$, $\text{UL}$, and $\modkl$.

To understand these classes, recall the notion of a non-deterministic Turing machine.
Such a machine may, at each step, choose to follow one or more computational paths.
For a "yes" instance, we just require that at least one of these paths be accepted.
This contrasts with a deterministic machine, which follows exactly one path.
For example, consider the directed graph connectivity problem:

\vspace{0.3cm}
\noindent \textbf{DAG}
\begin{itemize*}
    \item Input: A directed acyclic graph $G$, and a designation of two nodes in the graph, called $s$ and $t$.
    \item Output: $1$ if there exists at least one path from $s$ to $t$ in $G$, $0$ otherwise. 
\end{itemize*}
Starting at node $s$, a non-deterministic machine can solve $\text{DAG}$ by following every outward edge from $s$, and every outward edge from each subsequent node, etc.
The machine accepts if any of these computational branches reaches $t$.
We can restrict the computational power of the machine by requiring each branch, separately, run in a restricted amount of time or use a restricted amount of memory.

$\NL$ is the class of decision problems solvable on a non-deterministic Turing machine with $O(\log n)$ memory, where $n$ is the length of the input.
$\text{UL}$ is the class of decision problems solvable on a non-deterministic Turing machine with logarithmic memory, but requiring that exactly one branch accept on ``yes'' instances, and zero branches accept on ``no'' instances.
Finally, recall that $\text{L}$ is the class of decision problems that can be decided in $O(\log n)$ space on a deterministic Turing machine.
It is clear that $\text{L} \subseteq \text{UL}$, because a deterministic machine is a special case of a non-deterministic one, and the deterministic machine has just one computation path, and so in particular one accepting path.
 
It's also immediate that $\text{UL} \subseteq \NL$, because machines with one accepting path are special cases of the general non-deterministic one.

Finally, we consider $\modkl$, for $p$ prime.
This has an unusual definition, but turns out to capture the complexity of a number of natural problems.
$\modkl$ is the class of decision problems which can be solved by running a non-deterministic Turing machine and outputting "yes" when the number of accepting paths in that machine is non-zero mod $p$, and outputting "no" otherwise.
An example of a problem in this class is the following.

\vspace{0.3cm}
\noindent \textbf{DAG$_p$}
\begin{itemize*}
    \item Input: A directed acyclic graph $G$, and a designation of two nodes in the graph, called $s$ and $t$.
    \item Output: $1$ if the number of distinct paths from $s$ to $t$ in $G$ is non-zero mod $p$, and $0$ otherwise. 
\end{itemize*}
More relevantly, Ref.~\cite{buntrock1992structure} proved that $\modkl$ includes many natural linear algebra questions over the field $\mathbb{Z}_p$, including inverting and powering matrices, calculating the rank of a matrix and others.
To relate this to our earlier classes, note that a $\text{UL}$ machine on "yes" instances has one accepting path, so in particular $1$ mod $p$ accepting paths, so any problem in $\text{UL}$ can be decided in $\modkl$ so that $\text{UL} \subseteq \modkl$.
Together with $\text{L}\subseteq \text{UL}$, this also implies that $\text{L}\subseteq \modkl$ as mentioned in the introduction.

In Ref.~\cite{karchmer1993span}, it was pointed out that running a span program of polynomial size is in $\modkl$, and in fact every problem in $\modkl$ can be reduced in an efficient way to running a span program.
Consequently, we have
$$\PSP_p = \modkl \label{eq:psp}.$$
As a consequence of this, it is also true that a span program with $d$ rows can be computed by running a Turing machine with $O(\log d)$ memory, and outputting $0$ iff the number of accepting paths is non-zero mod $p$.

Using this, we can relate the classes $\text{L}$ and $\PSP_p$ according to
\begin{align}
    \text{L} \subseteq \text{UL} \subseteq \modkl = \PSP_p.
\end{align}
It is also believed that $\text{L}\subsetneq \NL$, and that $\text{UL}=\NL$.
Assuming both these statements, we would have that $\text{L} \subsetneq \PSP_p$.
We motivate these beliefs below.

First consider the claim $\text{L}\subsetneq \NL$. 
This is widely believed, similar to the belief that $\P \subsetneq \text{NP}$. 
It amounts to the statement that allowing a log space Turing machine to follow many computational paths at once adds power.
One line of evidence for $\text{L}\subsetneq \NL$ is the theory of NL-completeness.
Many problems \cite{jones1976new} are known to be NL-complete, meaning any problem in $\NL$ can be mapped to them using a log space mapping.
If $\text{L}$ is equal to $\NL$, then all of these problems have a log space solution, but no such solution is known for any of them.
Concretely, the DAG problem described above is $\NL$-complete.
This means the claim that $\text{L}\subsetneq \NL$ amounts to the statement that we cannot solve this problem in log space without non-determinism.

The second claim is that $\text{UL}=\NL$.
As mentioned above, it is immediate that $\text{UL} \subseteq \NL$, so it remains to understand the evidence for $\NL \subseteq \text{UL}$.
This was discussed in Ref.~\cite{reinhardt2000making,allender1999isolation}, where they pose the question in terms of the DAG problem. We summarize their argument briefly. 
First notice that since DAG is $\NL$-complete, if we can show it is in $\text{UL}$ we are done.
The problem then is to, given a directed graph $G$, define a non-deterministic Turing machine $M$ that has exactly one accepting computational path when there are any number $P\geq 1$ of paths in $G$ from $s$ to $t$, and no accepting computational paths otherwise.
It is not known how to solve this problem in this form.
However, consider rather than a $\text{UL}$ machine, a $\text{UL}$ machine which additionally has access to an advice string, which here will be a list of randomized weightings assigned to the edges of $G$.
Then, one uses that after assigning random weightings to the edges with high probability there will be a unique minimal weight path in $G$ from $s$ to $t$.
We build the machine $M$ to only accept on this minimal weight path, which gives it a single accepting computational path.
 
We can modify this construction to ensure it works with probability one. 
In particular, there exists a log-space computable function which maps from the advice string and the graph $G$ to a set of $n^2$ graphs $G_i$, each of which is a weighted version of $G$, such that for any graph $G$ at least one of the $G_i$ has a unique minimal weight path.
By exploiting the uniqueness of this path, one can solve DAG in $\text{UL}$.
The reader should refer to Ref.~\cite{reinhardt2000making} for more details.
 
It remains to remove the need for the $\text{UL}$ machine to access the advice string.
In Ref.~\cite{allender1999isolation}, it was shown that this can be done if suitable pseudo-random functions exist.
A pseudo-random function is one whose outputs are hard, in a suitable sense, to distinguish from completely random outputs.
In particular it is thought that there are pseudo-random functions that are much easier to compute than they are to distinguish from randomness.
In the construction above, we used an advice string assigning random weights to the edges in $G$.
We consider replacing this with an assignment by a pseudo-random function $p(x)$ which is computable in log space.
This assignment can be made by our $\text{UL}$ machine.
Then either there is a $p(x)$ which will create a graph $G_i$ with a unique minimal weight path, or distinguishing $p(x)$ from a truly random one is no harder than checking that all the $G_i$ have non-unique minimal weighted paths.
Given what is believed about pseudo-random functions, checking if the $G_i$ have unique minimal weight paths would too easily distinguish $p(x)$ from random, so we expect there is a log-space computable function that assigns suitable weightings.
From this we conclude that $\NL=\text{UL}$.

\vspace{0.3cm}
\noindent \textbf{L$_{(2)}$ and $\PSP_{p,(2)}$}
\vspace{0.3cm} 

In the last section we gave evidence, based on the existence of suitable pseudorandom functions, that $\text{L}\subsetneq \PSP_p$. 
Unfortunately, we cannot offer similar evidence separating $\text{L}_{(2)}$ and $\PSP_{p,(2)}$, although we believe this is the case.
More generally, for any classes A, B such that $\text{A}\subsetneq \text{B}$ it is unclear when $\text{A}_{(2)}\subsetneq \text{B}_{(2)}$.
We offer only some comments on this problem.

To understand this separation problem better, first of all consider some cases where A and B \emph{do} collapse under local pre-processing.
Trivial examples occur whenever one of two conditions are met.
If there is a promise that the inputs are of the form $(x,x)$, so that both local pre-processors see the full input, then the pre-processed classes A$_{(2)}$ and B$_{(2)}$ both become equal to the set of all functions, since we can have $\alpha$ or $\beta$ carry out the entire computation.
Another collapse occurs when the class B is defined by taking A and allowing for an advice string.
In that case having $\alpha(x)=(x,a)$ for $a$ the advice string and $\beta(y)=y$ collapses the classes.
For example\footnote{Recall that $L / poly$ is the class of functions computable in log-space with access to a polynomial size advice string.}, $\text{L} \subsetneq \text{L} / poly$ but this reasoning shows L$_{(2)}= \text{L}/ poly_{(2)}$.
Our example of $\text{A}=\text{L}$ and $\text{B}=\modkl$ does not have either of these features, so at the very least it cannot be obviously collapsed in either of these ways.

Another observation is that, when allowing arbitrary pre-processing, \emph{all} functions are contained in PSPACE$_{(2)}$. 
To see why, take $\alpha(x) = (x,f(x,y_1),...,f(x,y_{2^n}))$ and $\beta(y)=y$. 
Then, the local processor need only look up the $y$th element of the string $f(x,y_1),...,f(x,y_{2^n})$ and output the corresponding bit, and this can be done in $PSPACE$. 
This means for example that $PSPACE\subseteq EXP$ which is believed strict, but $PSPACE_{(2)}= EXP_{(2)}$. 
Because our classes $\modkl$ and $L$ are so much weaker than $PSPACE$, we do not believe a collapse by any similar mechanism is plausible in our case. 
 
To argue that a maintained separation under pre-processing is at least possible for some classes A and B, we prove such a separation in other cases.
Such separations are easy to prove for some low-lying complexity classes using tools from \emph{communication complexity}.
To define communication complexity, consider the following scenario.
Alice is given a string $x$, and Bob a string $y$.
Alice and Bob will communicate by sending classical bits to one another with the goal of determining the output of some Boolean function $f(x,y)$.
Unlike in a non-local computation scenario, they can communicate over many rounds.
Alice sends Bob a message, then, conditioned on the message he receives, Bob sends Alice a message, etc.
The communication complexity is then the total amount of information transferred from Bob to Alice plus the information sent from Alice to Bob. See Ref.~\cite{kushilevitz1997communication} for an introduction to communication complexity.

To understand why communication complexity can be used to separate classes with pre-processing, we first need to define the notion of a \emph{decision tree}.
A decision tree defines a simple type of program for computing a Boolean function on $n$ bits.
It consists of a directed tree\footnote{Recall that in graph theory, a directed tree is a directed acyclic graph whose underlying undirected graph is a tree, while a tree is an acyclic connected undirected graph.} such that except for the leaves and one other vertex specified as the \emph{root}, every vertex has one edge in and two edges out; a set of queries $Q$ consisting of functions of $O(1)$ input bits; a query $q_v\in Q$ for each non-leaf vertex $v$ in the graph; and a label for each leaf as either $0$ or $1$.
Starting at the root, for each node $v$ in the tree, the program checks the corresponding $q_v$ of the inputs.
Based on if that condition is true or false, it moves to the left or right branch from the current node.
Eventually the program reaches a leaf of the graph, and outputs the label of that leaf.
 
Decision tree size is related to communication complexity via the bound \cite{nisan1993communication}
\begin{align}\label{eq:ccbound}
    \dt(f(x,y)) \geq D(f(x,y)) / c_Q
\end{align}
where $D(f(x,y))$ is the communication complexity of the function $f(x,y)$, and $\dt(f(x,y))$ is the minimal depth of a decision tree computing $f$ using the set of queries $Q$.
The constant $c_Q$ is defined by $c_Q = \max_{q\in Q} D(q)$, the communication complexity of an individual query in the worst case. Briefly, this bound holds because a decision tree can be converted into a communication protocol: starting at the root, Alice and Bob communicate to evaluate the first query. This has communication cost at most $c_Q$. Given the output from this query, they follow the decision tree to the next node, and carry out another communication protocol to evaluate the next query. The communication cost is at most $c_Q$ times the depth of the tree $\dt(f(x,y))$, and this bounds the cost of the best possible protocol $D(f(x,y))$ from above. 

Define the complexity class $\DTq(F(n))$, consisting of problems solvable using decision trees with depth $O(F(n))$, and using queries $q$ drawn from some set $Q$. 
We claim that $\DTq(\sqrt{n})\subsetneq \DTq(n)$, and that $\DTq_{(2)}(\sqrt{n}) \subsetneq \DTq_{(2)}(n)$. 
We take the set of queries to be any relation on $O(1)$ inputs, in which case $c_Q=O(1)$. 
To show the first separation, consider the disjointness function
\begin{align}
    f_{disj}(x,y) = \begin{cases}
        1 \,\,\,\,\, \forall i, x_i\wedge y_i = 0 \\
        0 \,\,\,\,\, \text{otherwise}
    \end{cases}.
\end{align}
This has an obvious decision tree of size $n$: each node $n_i$ checks $x_i\wedge y_i$, with the output from that node labelled $0$ going to a leaf labelled $0$, and the output from $n_i$ labelled $1$ mapping to node $n_{i+1}$. 
This shows $f_{disj}(x,y)\in \DTq(n)$. 
As well, it is easy to show using lower bounds on communication complexity that $D(f_{disj}(x,y))\geq n$, so from the bound \ref{eq:ccbound} we get that $f_{disj}(x,y)\not\in \DTq(\sqrt{n})$, separating the two classes. 

Finally, we show the separation between the corresponding locally pre-processed classes. First, note that $f_{disj}(x,y)\in \DTq_{(2)}(n)$, since it is in the smaller class $\DTq(n)$. Next, suppose by way of contradiction that $f_{disj}(x,y) \in \DTq_{(2)}(\sqrt{n})$. Then there exists a function $F\in \DTq(\sqrt{n})$ such that $f_{disj}(x,y) = F(\alpha(x),\beta(y))$. But then
\begin{align}
    n \leq D(f_{disj})\leq D(F)
\end{align}
where the first inequality we mentioned above and is easy to prove in communication complexity, and the second inequality is immediate, because the definition of communication complexity allows for local pre-processing with arbitrary functions. Using \cref{eq:ccbound} and $f_{disj}(x,y) \in \DTq_{(2)}(\sqrt{n})$, we have
\begin{align}
    D(F)\leq c_Q \, \dt(F) \leq O(\sqrt{n})
\end{align}
which is a contradiction, so there is no such function $F$. This shows $f_{disj}\not\in \DTq_{(2)}(\sqrt{n})$, so $\DTq_{(2)}(\sqrt{n}) \subsetneq \DTq_{(2)}(n)$. 

While the strategy used above is natural to apply to our notion of local pre-processing, it cannot be applied to the classes $\text{L}$ and $\modkl$. This is because $\text{L}$ includes problems which require super-linear decision trees, and $D(f)\leq 2n$ always.\footnote{Using $2n$ bits of communication, Alice and Bob can send each other their full input strings.} This means we cannot hope to separate $\text{L}$ from a larger class using the bound \ref{eq:ccbound}. The technique does generalize to separate $\DTq$ classes of size less than $n$ however, by finding a function with suitable communication complexity, which can always be found.\footnote{For example, the disjointness function on some portion of the inputs of size $f(n)$ has communication complexity $f(n)$.} At least for these classes then, adding more computation power to the local computation makes the pre-processed classes larger. Our code-routing protocol improves on the garden-hose strategy if this remains true for the larger classes $\text{L}$ and $\modkl$. Understanding this for these or other classes however appears challenging, and we have not encountered any techniques for doing so which apply to $\text{L}$ and $\modkl$. 

\subsection{Upper bounds on efficiently achievable complexity}\label{sec:converse}

\Cref{thm:codeprotocol} lower bounds the complexity of functions that can be completed using code-routing protocols, showing it completes the routing task non-locally at least for functions in $\modkl_{(2)}$, when restricted to polynomial entanglement.
The protocol used to establish this is a restricted one however, and it is natural to ask if the more general procedure can complete functions of higher complexity.
 To increase the power of the code-routing strategy, we could:
\begin{itemize}
    \item \emph{Use other codes.} The codes we used that arise from Smith's construction \cite{smith2000quantum} (``Smith codes''), are CSS codes,\footnote{We have not found this statement in the literature but it is easy to verify.
    In fact, every CSS code is also a Smith code, as we discuss in an upcoming work.
    } so it is clear they are a restrictive set. 
    \item \emph{Unit-route on predominantly locally-held bits.} If most unit-routing is done on bits held by the other player, then the entanglement cost from the necessary teleportations is closely related to the total share size of the codes used. But by unit-routing many shares on locally-held bits, the total share size may not capture the entanglement cost. 
    \item \emph{Use adaptive encoding.} To prove \cref{thm:codeprotocol}, we used a single, fixed encoding on Alice$_0$'s side. More generally, which encoding is performed can depend on the classical inputs. As well, shares teleported to Alice$_1$'s side could be themselves encoded, shares from those teleported back and encoded, etc. 
\end{itemize}
We are not able to fully characterize the complexity of functions that can be achieved with polynomial entanglement using a general combination of the above strategies.
We are able however to give a few partial results.
To phrase our results, it is helpful to have a notion of size for a protocol. 
The protocol tape $I$ for a given set of inputs $(x,y)$ (see \cref{def:crprotocol}), defines a pattern of encoding that we refer to as the \emph{protocol tree}.
Each $S_i$ defines a vertex in a directed tree with inputs $v_i$ and outputs $\{w_i^j\}$.    
We define the \emph{size} of a protocol tree as the number of leaves, plus the number of internal wires that correspond to teleportations. To count this, it is helpful to define $n_k\equiv |\{w_k^j\}|$. Then we define the size of a protocol tree as
\begin{align}\label{eq:sizedef}
    H_{(x,y)} \equiv \left(1 + \sum_{k:n_k > 1} (n_k - 1) + \sum_{k:n_k = 1} 1 \right) 
\end{align}
The protocol size counts the \emph{number of shares} which are either unit-routed or teleported. 
This lower bounds another quantity of interest, which is the total log dimension of all the shares either unit-routed or teleported during the protocol, which we call the \emph{weighted protocol tree size} and denote $\tilde{H}_{(x,y)}$. 
To count this, it is helpful to define $\tilde{n}_k=\sum_i \log \dim w_k^i$. Then we have
\begin{align}\label{eq:weightedsizedef}
    \tilde{H}_{(x,y)} \equiv \left(\log \dim Q + \sum_{k:n_k > 1} (\tilde{n}_k - \log \dim v_k) + \sum_{k:n_k = 1} \tilde{n}_k \right) .
\end{align}
If a share is unit-routed on a bit that is on the same side as the share, there is zero entanglement cost, while if the share is on the opposite side, there is an entanglement cost given by the log dimension of the share. 
Each share which is teleported gives an entanglement cost equal to the log dimension of that share. 
Our assumption in the converse results below will be that a polynomial in the entanglement cost upper bounds the weighted protocol tree size $\tilde{H}_{(x,y)} \leq \text{poly}(E)$. This is our precise statement of not too many unit-routings being performed on locally held bits. 

We begin with the following theorem, which shows code-routing using Smith codes is in $\P_{(2)}$, under our assumption relating protocol tree size and entanglement cost. We can also strengthen this to $\modkl_{(2)}$ if the protocol tree is $O(1)$ depth, or $L_{(2)}$ if each encoding has $O(1)$ size. 
Theorems \ref{thm:Pconverse} and \ref{thm:complexityupperbound-wide} also have alternative proofs in terms of composed span programs, which we haven't included here. 

\begin{theorem}\label{thm:Pconverse}
Consider a code-routing protocol which uses only Smith codes, uses $E=\text{poly}(n)$ copies of the maximally entangled state of two qupits, and has protocol trees with size related polynomially to their entanglement cost. Then we can determine the outcome of the protocol in $\P_{(2)}$, polynomial time with local pre-processing.
\end{theorem}
\begin{proof}
\,\,We will give an explicit $poly(E)$ time algorithm.
Recall that the protocol tape consists of a list
\begin{align}
    I=(a(x),S_1,...,S_\ell,b(y),S_{\ell+1},...,S_{\ell+\ell'})
\end{align}
and each $S_i=(v_i,\{w_i^j\},T_i)$ describes a unit-routing, teleportation, or encoding.
By assumption, the encoding here corresponds to a Smith code.
It will be convenient in this proof to take $T_i$ to be a description of the span program defining that Smith code.
To denote this, when the third entry describes an encoding, we will use the labelling $SP_i$ rather than $T_i$, i.e.\ $S_i=(v_i,\{w_i^j\},SP_i)$.
Recall also that the size of the span program is equal to the number of rows in its matrix. 

Given this representation of the protocol, we define the following recursive function which takes a tuple $S_k$ as input and returns $0$ if Alice$_0$ is able to reconstruct the input share $v_k$, or returns $1$ if Alice$_1$ is able to reconstruct the share $v_k$. 
In the pseudo-code below, we denote a span program by $SP_k$, where each span program is defined by a tuple $SP_k=(M_k,\phi_k,t_k)$, where function $\phi_k$ maps from a row index $i$ to a pair $(j,\epsilon_i)$, as explained in \cref{appendix:spanprograms}.
We use the notation $\phi_k(i)[1]=j$.
Note that Smith codes are defined by monotone span programs, meaning that $\epsilon_i=1$ always.\\

\noindent Define $\text{GetOwner}(S_k,I)$:\\
\indent If $n_k=0$,\\
\indent \indent Return $T_k$  \\
\indent If $n_k=1$, \\
\indent \indent Search for $S_i\in I$ with $w_i^0$ as its input, call it $S_j$ \\
\indent \indent Return $\text{GetOwner}(S_j,I)$ \\
\indent $M^1_{k}=\{\}$\\
\indent For $i$ from $1$ to $\size(SP_p)$,\\
\indent \indent Set $v = \phi_k(r_i)[1]$ \\
\indent \indent Search for $S_i\in I$ with $w_i^v$ as its input, call it $S_v$\\
\indent \indent If $\text{GetOwner}(S_v,I) = 1$, \\
\indent \indent \indent Append $r_i$ to $M^1_k$ \\
\indent If $t_k \in \spn{M}^1_k$, \\
\indent \indent Return 1 \\
\indent Else, \\
\indent \indent Return 0. \\

Then, our program is as follows:\\

\noindent Find the tuple $S_i$ with $Q$ as its input, call it $S_k$\\
\noindent Return $\text{GetOwner}(S_k,I)$\\

It is straightforward to see that this algorithm is correct using an inductive proof, where we induct on layers in the protocol tree.
Here, we say that the layer of a node is the maximal length of a path from that node to a leaf.
The $0$th layer -- the leaves of the tree -- all correspond to unit-routings, where the algorithm is manifestly correct: unit-routings have $n_k=0$, and $T_k$ is a bit labelling the side that the input share is brought to in the protocol.
The algorithm just returns this bit directly, which is correct.
Now assume by way of induction that the algorithm behaves correctly on tuples $S_{k'}$ at layer $m$ of the protocol tree, and consider its behaviour on a tuple $S_k$ at the $m+1$th layer.
We have that $n_k \neq 0$, so we need only consider the cases where $n_k=1$ or $n_k>1$. 

For $n_k=1$ the protocol has teleported $v_k$ into system $w_i^0$, which is in the $m$th layer, so the algorithm returns the side where $w_i^0$ is brought, which is correct. 

For $n_k>1$, the share $v_k$ has been recorded into a secret sharing scheme. The scheme is defined by a span program, and records $v_k$ into a set of shares $\{w_k^i\}$. The scheme's indicator function is computed by a monotone span program $(M_k,\phi_k,t_k)$. The share $v_k$ will be recoverable on the side labelled by the output of the span program. The inputs to the span program $z_i$ are determined by where the protocol brings the shares $w_k^i$, with $z_i=0$ meaning share $w_k^i$ is on the left and $z_i=1$ meaning share $w_k^i$ is on the right. Share $v_k$ is then available on the side labelled by the indicator function evaluated on the string $z$. The algorithm works by evaluating the span program, and calling the $\text{GetOwner}(\cdot,I)$ function recursively to determine on which side the shares $w_k^i$ are recoverable. In particular the matrix $M_k^1$ includes $t_k$ in its span exactly when the span program evaluates to $1$, so the algorithm correctly returns $1$ when $v_k$ is on the right. When the set of shares on the right does not reveal $v_k$ it must, because we used a secret sharing scheme, reveal nothing about $v_k$. Because we always maintain the purifying system on the left, $v_k$ is then available on the left. Accordingly, the algorithm correctly returns $0$ in this case. 

Next we analyze how the run time relates to the entanglement cost.
Begin by considering the run time for each call to $\text{GetOwner}(S_k,I)$.
The run time is dominated by the step where we determine whether an $e$-dimensional vector $t_k$ lies in the span of another set of $|M_{k}^1|$ vectors.
This can be done in $O(e|M_{k}^1|)$ steps.
The length of the rows is always less than or equal to the total number of them, since the columns are linearly independent\footnote{This follows because any column expressible as a linear combination of other columns amounts to a redundant condition on the requirement for a set of rows to have the target vector in its span; thus it can be safely deleted from the span program matrix without changing the function that the span program computes.}, so $e\leq \size(SP_k)$.
The number of rows in $M_k^1$ is less than or equal to the total number of rows in the span program, so $|M_k^1|\leq \size(SP_k)$.
Together these give $O(e|M_{k}|)< O(\size(SP_k)^2)$. 
In a Smith code, the total share size is given by the size of the span program, so $\tilde{n}_k = \size(SP_k)$.
Finally, note that on a given input pair $(x,y)$ only certain span programs from the full collection $\{S_k\}$ are reached in the algorithm. 
Call this collection $\mathcal{S}_{(x,y)}$. Thus we can bound the total run time for a given $x$ and $y$ by
\begin{align}\label{eq:timebound}
    T_{(x,y)} \leq \sum_{k\in \mathcal{S}_{(x,y)}} \tilde{n}_k^2 \leq N_{(x,y)}^2
\end{align}
where $N_{(x,y)}=\sum_{k} \tilde{n}_k$ is the total size of all shares used across all encodings involved in the protocol, on inputs $(x,y)$.
We would like to relate this run time to the protocol tree size, as defined in \cref{eq:weightedsizedef}. For fixed $N_{(x,y)}$, the weighted protocol tree size is minimized for the case where $n_k=2$ for all encodings (this maximizes the subtractions appearing in \cref{eq:weightedsizedef}), so that
\begin{align}
    \tilde{H}_{(x,y)}\geq N_{(x,y)}/2
\end{align}
where we've also used that $\frac{\tilde{n}_k}{2} \geq \log \dim v_k$, i.e. that each share in the code is at least as large as the input system. 
Since by assumption the entanglement cost is polynomially related to the weighted size, combining this with \cref{eq:timebound} we have a polynomial upper bound on the run time in terms of entanglement cost. Note that this polynomial time computation is performed by taking the protocol tape as input, which itself is computed via local pre-processing, so the entire protocol is in $P_{(2)}$.
\end{proof}

For certain classes of code-routing protocols, we can determine their output in smaller classes than $\P_{(2)}$. This is possible in two cases: protocols which never concatenate codes to depth more than $O(1)$, and protocols which are built by concatenating codes of $O(1)$ size. We can understand the first of these as a small relaxation of the single-encoding protocol given in \cref{thm:codeprotocol}, and the second as a small relaxation of the garden-hose protocol. In both cases deforming these protocols slightly doesn't add computational power. We discuss these two cases in the following subsections. 

\subsubsection*{Protocols using $O(1)$ depth encodings}

We first discuss the following theorem, which modifies the protocol used in \cref{thm:codeprotocol} to allow $O(1)$ depth of encodings and shows the resulting protocols still compute functions inside the class $\modkl$.

\begin{theorem}\label{thm:complexityupperbound-wide}
Consider a code-routing protocol which uses only Smith codes, takes $n$ bits as input, uses $E=\text{poly}(n)$ copies of the maximally entangled state of two qupits, has protocol trees with size related polynomially to their entanglement cost and which have $O(1)$ depth. Then the outcome of the protocol can be computed in $\modkl_{(2)}$. 
\end{theorem}

Our proof will use the following characterization of $\modkl$ in terms of non-deterministic Turing machines.
For any non-deterministic Turing machine $T$ we define the function $\mathcal{F}(T)$ as follows. 
For a given input $x$, call the number of accepting paths $F(x)$. 
We then define $\mathcal{F}(T)(x)=1$ when $F(x)$ is non-zero mod $p$, and return $\mathcal{F}(T)(x)=0$ otherwise. 
Then the class $\modkl$ is the set of functions of the form $f=\mathcal{F}(T)$ where $T$ has $O(\log(n))$ memory for $n$ the length of $x$. 
Note that because Smith codes of polynomial size are evaluated by polynomial sized span programs, and hence in $PSP_p$, and recalling that $\modkl=PSP_p$ \cite{karchmer1993span}, we have that they can also be evaluated by non-deterministic Turing machines with $O(\log (n))$ memory that count paths mod $p$.

To prove \cref{thm:complexityupperbound-wide}, we first need the following lemma, which will allow us to compose $\modkl$ machines in a simple way. 

\begin{lemma}\label{lemma:paths}
Suppose we have a function $f = \mathcal{F}(T)$ for a non-deterministic Turing machine $T$ running on memory $m=  \Omega(\log n)$ where $n$ is the length of $x$. 
Then there is another non-deterministic Turing machine $T'$ that uses memory $O(m)$, has $f(x)$ mod $p$ accepting paths (and therefore still satisfies $f = \mathcal{F}(T')$), and has $1-f(x)$ (mod $p$) rejecting paths.
\end{lemma}
\begin{proof}\,
We will start with any Turing machine $M_0$ such that $f=\mathcal{F}(M_0)$, and from it construct a new machine $M_2$ whose number of accepting and rejecting paths will satisfy the statement of the lemma.
As an intermediary, we need another Turing machine $M_1$. We will use $F_i(x)$ to denote the number of accepting paths in Turing machine $M_i$ run on input $x$, and $\bar{F}_i(x)$ the number of rejecting paths. \\

The machine $M_1$ uses $p-1$ copies of $M_0$, which we label $M_0^{(i)}$ with $i\in\{1,...,p-1\}$. It is defined as follows. \\

\noindent Define $M_1$: \\
\indent For $i\in\{1,...,p-1\}$ \\
\indent \indent Run $M_0^{(i)}$ \\
\indent \indent If $M_0^{(i)}$ is in reject state, \\
\indent \indent \indent Reject\\
\indent Accept\\

$M_1$ runs $p-1$ copies of $M_0$, and accepts only if all $p-1$ copies enter accept states.
Consequently, the number of accepting paths is
\begin{align}
    F_1(x) &= (F_0(x))^{p-1} \nonumber \\
    &= f(x) \,\, (\text{mod}\,\, p)
\end{align}
where in the second line we've used Fermat's little theorem. Next, we build the machine $M_2$. \\

\noindent Define $M_2$:\\
\indent Goto both the next two lines\\
\indent Reject\\
\indent Run $M_1$ \\
\indent If $M_1$ is in accept state,\\
\indent \indent Non-deterministically pick $j \in\{0,...,p-1 \}$ \\
\indent \indent If $j  > 0$,\\
\indent \indent \indent Reject \\
\indent \indent Accept\\
\indent If $M_1$ is in reject state,\\
\indent \indent Non-deterministically pick $j \in \{0,...,p-1\}$\\
\indent \indent Reject\\

$M_2$ has the same number of accepting paths as $M_1$, which is $f(x)$ mod $p$. For the rejecting paths, we have $p-1$ paths introduced for each accept path of $M_1$, plus $p$ additional paths from each reject state, plus one additional path from the first line. So the number of rejecting paths of $M_2$ is given by
\begin{align}
    \bar{F}_2(x) &= 1 + (p-1)F_1(x) + p\bar{F}_1(x) \nonumber \\
    &= 1 - f(x) \,\, (\text{mod}\, \, p)
\end{align}
as needed.

Notice that $M_2$ involves running $M_0$ an $O(1)$ number of times sequentially, storing $j $, and keeping track of the $i$ counter. All this can be done in $O(m)$ memory. 
\end{proof}

\vspace{0.1cm}
Now we are ready to prove the main theorem of this section. 
\vspace{0.1cm}

\begin{proof} \textbf{\,(Of theorem \ref{thm:complexityupperbound-wide})}\,
We use the description of the protocol in terms of a protocol tape. 

Recall that when $S_i$ has no output shares, the tuple $S_i=(v_i,\emptyset,T_i)$ describes a unit-routing of the share $v_i$ to the side labelled by $z_{T_i}$, which is a bit of $z=(a(x),b(y))$.

When $S_i$ has one output share, $S_i=(v_i,w_i^0,\emptyset)$ describes a teleportation.

Finally when $S_i$ has more than one output share, the tuple describes an encoding. The encoding is into a Smith code, so the indicator function $f_i$ can be computed with a span program of size $\tilde{n_i}$. 
To find a Turing machine such that $f_i=\mathcal{F}(T)$, we need only memory $O(\log \tilde{n}_i)$. 
From lemma \ref{lemma:paths} then, we can construct a non-deterministic Turing machine $T_i$, also with memory $O(\log \tilde{n}_i)$, such that $T_i$ has $f_i(x)$ mod $p$ accepting paths and $1-f_i(x)$ mod $p$ rejecting paths. 

We consider a function $L(s,I)$, which takes a share $v$ and determines if that share is on the left (corresponding to output $0$) or the right (corresponding to output of $1$) at the end of the protocol defined by input tape $I$.
We define $L(s,I)$ recursively, as follows. \\

\noindent Define $L(s,I)$:\\
\indent Search through $I$ and find $S_i$ with $s=v_i$
\\
\indent If $n_i=0$,\\
\indent \indent Return $z_{T_i}$ \\
\indent If $n_i=1$, \\
\indent \indent	Return $L(w_i^0,I)$ \\
\indent Else,\\
\indent \indent Return $f_i(L(w_i^0,I),...,L(w_i^{n_i},I)))$\\

Note that this machine does not compute each of the $L(w_i^j,I)$ and store them --- that would already be $n_i$ bits of memory.
Instead it computes $L(w_i^j,I)$ each time it needs that bit value, and can re-use the same memory bits each time it does this.
The output of the entire protocol is determined by running $L(Q,I)$, where $Q$ is the input system to be routed. 

$L(Q,I)$ determines the output for the protocol, but we need to show this function can be evaluated by a $\modkl$ machine. 
To do so, we modify $L(s,I)$ to a new function $L_T(s,I)$ by making the replacement $f_i\rightarrow T_i$, where $T_i$ is a Turing machine constructed using \cref{lemma:paths}. 
$L_T(Q,I)$ can be run on a non-deterministic machine, and we can consider counting the number of accepting paths. 
Our claim is that 1) this correctly determines the output of the protocol in that $\mathcal{F}(L_T(s,I))=L(s,I)$ and 2) $L_T(Q,I)$ runs in non-deterministic log-space, so that we've computed the output of the protocol in $\modkl$. 

First consider correctness. We work inductively in the layers of the protocol tree, where the layer of a node is defined as before to be the maximal length of a path from the node to a leaf.
We will show for each layer that, for any node in that layer, the number of accepting paths is equal, mod $p$, to the output of the corresponding function and further that the number of rejecting paths is equal, mod $p$, to $1$ minus the value of that function. 

First consider the $0$th layer, i.e.\ the leaves of the tree, which will always consist of unit-routings.
These are deterministic computations, consisting of returning $z_{T_i}$ (which in this case is a single bit). They return $z_{T_i}$ if and only if there is $z_{T_i}$ accepting paths, and have $1-z_{T_i}$ rejecting paths, so this is correct. 

Next consider the $m+1$th layer of the protocol tree, and assume the inductive hypothesis for the $m$th layer.
For an encoding, to evaluate the function $f_i$ on a log-space machine we need non-determinism.
Consider the function $f_i$, its corresponding Turing machine $T_i$, and focus on one input to $f_i$, say $z_*$. 
By construction, for a definite input (or a single path) with $z_*=z$, we know $T_i$ has $f_i(z)$ accepting paths and $1-f_i(z)$ rejecting paths. 
Now suppose we replace the input $z_*$ with calls to a non-deterministic Turing machine $T_*$ at the $m$th layer.
Then including all input paths from $T_*$ as well as all paths for $T_i$ itself, the number of accepting paths for $T_i$ is the number of accepting paths for $T_i$ given $z_*=1$, times the number of accepting paths for $T_*$, plus the number of accepting paths for $T_i$ given $z_*=0$, times the number of rejecting paths for $T_*$.
Using that the number of accepting paths of $T_i$ is $f_i(z_*)$, and rejecting paths is $1-f_i(z_*)$, and a similar statement for $T_*$ and associated function $f_*=\mathcal{F}(T_*)$, we have that the number of accepting paths for $T_i$ is
\begin{align}
F_i &= f_i(1) f_* + f_i(0) (1-f_*) \,\,(\text{mod}\,\, p).
\end{align}
Notice that for $z_*=f_*=1 \,\, \text{mod}\,\, p$, we have $F_i=f_i(1)$, so the number of accepting paths is as if $z_*$ were given deterministically. 
Similarly if $z_*=f_*=0$, $F_i=f_i(0)$, which again is the same as if $z_*$ were given deterministically. 
In particular, the number of accepting paths satisfies the requirements of the inductive hypothesis.
The number of rejecting paths of $T_i$ is 
\begin{align}
    \bar{F}_i = (1-f_i(1))f_* + (1-f_i(0))(1-f_*),\,\,(\text{mod}\,\, p)
\end{align}
using similar reasoning to above.
Thus for $z_*=f_*=1$, we have that the number of rejecting paths is $1-f_i(1)$, and for $z_*=f_*=0$, we have that the number of rejecting paths is $1-f_i(0)$, so that the number of rejecting paths also satisfies the inductive hypothesis. 
This argument also gives correctness in the case of a teleportation, since teleportation is a special case of the above where $T_i$ is deterministic. 

Finally we need to determine the memory usage of this algorithm.
The needed memory is to evaluate the Turing machines at each layer, which each use $\log \tilde{n}_{i}$ memory, where $\tilde{n}_{i}$ is the log dimension of the output shares of tuple $S_i$.
Calling Turing machines recursively, we can re-use memory for machines at the same layer of recursion, but must add the memory requirements for machines at different layers. Adding $\log |\{S_i\}|$ bits of memory for the search through the list of the $S_i$, calling $L_T(Q,I)$ uses
\begin{align}\label{eq:memory}
	M_{(x,y)} = \max_{\text{paths} \, p} \sum_{i \in p} \log \tilde{n}_{i} + \log |\{S_i\}|
\end{align}
bits of memory. The second term is bounded by $\log \tilde{H}_{(x,y)}$ for $\tilde{H}_{(x,y)}$ the weighted size of the protocol tree, since each $S_i$ adds at least $1$ to the size of the protocol tree.
Finally, note that the length of the path is bounded by the depth of the protocol tree.
Then using our assumption that we have at most $O(1)$ depth, and because $\tilde{n}_i\leq \tilde{H}_{(x,y)}$, we have
\begin{align}
    M_{(x,y)} &\leq O(\log \tilde{H}_{(x,y)}).
\end{align}
Because $\tilde{H}_{(x,y)}$ is related polynomially to the entanglement cost, we've proven the theorem. 
\end{proof}

\subsubsection*{Protocols using codes of $O(1)$ size}

In this section we consider protocols that use only codes with $O(1)$ shares. Recall that the garden-hose protocol corresponds to the case where encodings are size $1$, and the efficiently computable functions in that case is the class $\L_{(2)}$. The following theorem shows that with small codes the complexity is not increased. Note that this is our only converse theorem where we do not restrict to Smith codes. 

\begin{theorem}\label{thm:complexityupperbound-deep}
Consider a code-routing protocol that takes $n$ bits as input, uses $E=\text{poly}(n)$ copies of the maximally entangled state of two qupits as a resource, has protocol trees with size related polynomially to the entanglement cost, and uses codes with at most $O(1)$ shares. Then the outcome of the protocol can be computed in $\L_{(2)}$. 
\end{theorem}

\begin{proof}
\,\,The strategy is to use a depth-first evaluation of the protocol tree, which recall is defined by the protocol tape $I$.
One apparent obstruction is that for deep trees, keeping track of a path from root to leaf can require linear memory.
To avoid this, we travel through the tree while only keeping the current, and sometimes proceeding or subsequent, vertices in memory.

Heuristically, our algorithm works by ``pruning'' the protocol tree, evaluating sub-trees and storing the ownership of shares corresponding to edges of the tree.
To store the full protocol tree would require too large of a memory, so instead we describe the pruned tree using the protocol tape $I$ along with a set $R$, which contains edges that ``over-ride'' the description of the tree given by $I$.
At any given point in the running of the algorithm, $R$ will only describe the ownership of vertices that neighbour the current vertex $v$ being evaluated.
Because the tree has vertices only with $O(1)$ degree, it is possible to store $R$ in logarithmic memory.
By repeatedly pruning the initial tree, eventually we are left with a trivial tree that points to the location of the input share. 

We give the pseudo-code for our algorithm now, then make a few comments on this code below.\\

\noindent $R=\{\}$ \\

\noindent Define IsLeaf$[v,I,R]$ \\
\indent If $R$ contains an $S_i$ with $v$ as input, \\
\indent\indent	Return 1 \\
\indent If $I$ contains an $S_i$ with $v$ as input, \\
\indent\indent If $|\{w_i^j\}| = 0$\\
\indent\indent \indent Return 1 \\
\indent Return 0 \\
	
\noindent Define $F[v_i,I,R]$\\
\indent If IsLeaf$[w_i^j,I,R]=1$ for all $j$,\\ 
\indent\indent Remove any $S_j'$ with inputs $w_i^j$ from $R$\\
\indent\indent $S_i' = (v_i,\emptyset,f_i)$\\
\indent\indent Append $S_i'$ to $R$\\
\indent\indent If there is a $v_k$ which has $v_i$ as a descendant, \\
\indent\indent \indent Erase $v_i$\\
\indent\indent \indent $F[v_k,I,R]$\\
\indent\indent Else, \\
\indent\indent \indent Return $f_i$\\
\indent Else, \\
\indent \indent Find the $w_i^j$ of maximal layer, call it $w_*$\\
\indent \indent	Erase $v_i$\\
\indent \indent	$F[w_*,I,R]$\\
	
\noindent Call $F[Q,I,R]$\\
	
In the definition of $F[v_i,I,R]$, the line which assigns $S_i'$ the value $(v_i,\emptyset,f_i)$ needs some explanation. According to our conventions, when the output systems are empty, the third entry in an $S_i$ tuple is just a bit. Here, we use the value of $f_i$. The inputs to $f_i$ are determined by the locations of the $w_i^j$ shares, but by construction we are in a case where these are easy to look up, since the $w_i^j$s are all leaves. Further, because the code sizes are all $O(1)$ here, this can be done in $O(1)$ memory. 

One other line that requires explanation is the one that finds a $w_*$ of maximal layer.
First, note that the layer of a node can be evaluated in log-space, because it amounts to determining the depth of the sub-tree defined by that node and all its descendants.
Second, the layer of each of the children of the current node can all be stored simultaneously, because (i) there are only $O(1)$ children, and (ii) the layer is bounded by the depth of the protocol tree, which is at most polynomial in $n$ by the assumptions of the theorem, and thus can be stored in $\log(n)$ bits.
	
To understand the correctness of the algorithm, we will make use of a notion of an \emph{effective protocol tree}. This is the tree as described by $R$ taken together with $I$, where $R$ is always `given priority'. In particular, if $v_i$ is an input to $S_i'\in R$ and $S_i \in I$, we use $S_i'$ when travelling to subsequent nodes in the tree. We define the \emph{effective size} to be the \emph{number of vertices} in the effective protocol tree.\footnote{Note that, unfortunately, this is not the same as the size of the effective protocol tree, using our earlier definition of size.}

We claim that the effective tree constructed during the running of the above algorithm evaluates to the same value as the original tree at every step. Further, effective size decreases every time the first If statement is called, and eventually reaches 1.

To see the first claim, consider that at the start of the algorithm $R=\{\}$, so the effective and original protocol trees agree, and so in particular give the same output.
Next, suppose that the effective and original protocol trees give the same output, and then consider how $R$ is edited during one evaluation of the code inside the first If statement of $F$.
This involves replacing $S_i$ with $S_i'$ which is a unit-routing that has the same output as $S_i$.
Manifestly this doesn't change the output.
Further, we remove the descendants of $S_i$, which are never visited in the new effective tree, so this also does not change the output. 

Now consider the second claim, that the effective tree becomes smaller and eventually reaches size one. 
Notice that we must reach the first If statement eventually, specifically after at most a number of calls to $F$ equal to the depth of the effective tree.
In particular each time the second Else statement is called, $F$ is called on a lower vertex in the effective tree. 
Once the call is to a vertex with only leaves as descendants, it goes to the first If statement. 
Next, notice that $S_i$ is replaced with $S_i'$ only when $S_i$ has descendants, and that by construction $S_i'$ is a leaf. 
Thus every such move decreases the effective size. 
Notice further that the algorithm can only end when reaching the single return statement. 
This happens when there is no node preceding the current one in the effective tree, so that the tree has size one. 
The algorithm then returns $f_i$ from the effective tree, which by the correctness property above is the output of the protocol tree.

Consider the memory usage of this algorithm. 
We evaluate indicator functions $f_i$ for $O(1)$ size codes, which can be done with $O(1)$ memory. 
Additionally, we need to keep track of the current node $v_i$, which can be done with $\log |\{S_i\}|$ memory. 
Notice that we have been careful to erase the record of the path followed to reach the current vertex, by erasing the stored $v_i$ value before calling $F$ on a new one, since storing this path would require super-logarithmic memory. 
Finally, we track the entries in $R$, which defines the effective tree.
We claim $R$ only ever contains $S_i$ which are all descendants of a single node, so storing $R$ only requires $O(1)$ memory. 
To see why this is the case, notice that because we travel to the node of maximal layer when traversing the tree, we visit nodes depth-first. 
This guarantees that once a vertex is added to $R$, we completely finish evaluating the ownership of its parent before proceeding to the next vertex, as we are already at the deepest part of the tree.

Considering all contributions listed in the last paragraph, memory cost is $O(\log |\{S_i\}|)$. This is upper bounded by $O(\log H_{(x,y)})$, since each $S_i$ adds at least $1$ to the protocol tree size. 
Then since $H_{(x,y)}\leq \tilde{H}_{(x,y)}$ and $\tilde{H}_{(x,y)}$ is upper bounded by a polynomial in $n$, we are done. 
\end{proof}

\section{Discussion}

The $f$-routing task is of practical relevance in the context of position verification, but also exhibits interesting relationships to complexity theory and secret sharing.
In particular, the garden-hose protocol uses entanglement controlled by the space complexity of $f$, and the code-routing strategy we introduce here has an entanglement cost upper bounded by span program size.
With regards to secret sharing, we showed the size of a secret sharing scheme with indicator function $f$ is lower bounded by the entanglement cost of performing the corresponding $f$-routing task. 

These connections to complexity and secret sharing emphasize the importance, and difficulty, of finding lower bounds on entanglement cost in $f$-routing. In particular, such lower bounds would strengthen the security of position verification schemes based on $f$-routing, and amount to lower bounds on span program size and the size of secret sharing schemes. In general, proving lower bounds on complexity is a challenging goal, and in the case of span programs there has been only limited success \cite{karchmer1993span}.\footnote{More success is possible when restricting to monotone span programs, see e.g.  \cite{robere2016exponential} for recent work, but monotone span programs are not the relevant computing model here.} 
Given this, we might not expect to prove strong lower bounds on entanglement cost. Alternatively, we could hope for conditional lower bounds based on complexity-theoretic assumptions, or for lower bounds stated in terms of some measure of the complexity of $f$. We leave exploring this further to future work.

Finally, note that this work introduces the use of error-correction in non-local quantum computation. By combining error-correction with the teleportation techniques of \cite{buhrman2013garden}, we increase the complexity of functions that can be computed non-locally (at least given our complexity-theoretic assumptions). It would be interesting to understand if error-correcting codes provide enhancements to other non-local computation protocols, for instance the one based on the Clifford+T gate set described in \cite{speelman2015instantaneous}.

\vspace{0.3cm}
\noindent \textbf{Acknowledgements}
\vspace{0.3cm}

We thank Adam Bouland, Kfir Dolev, Anirudh Krishna and Patrick Hayden for helpful discussions. 
AM is supported by the Simons Foundation It from Qubit collaboration, a PDF fellowship provided by Canada's National Science and Engineering Research council, and by Q-FARM. 
SC is supported by a graduate fellowship award from Knight-Hennessy Scholars at Stanford University.

\appendix

\section{Span programs}\label{appendix:spanprograms}

To express an arbitrary function $f$ as described in \cref{lemma:decomposition}, we first give the following definition.

\begin{definition}\label{def:spanprogram}
A \textbf{span program} over a field $\mathbb{Z}_p$ consists of a triple $S=(M, \phi, \mathbf{t})$, where $M$ is a $d\cross e$ matrix with entries in $\mathbb{Z}_p$, $\phi$ is a map from rows of $M$, labelled $r_i$, to pairs $(k,\varepsilon_i)$, with $k\in \{1,...,n\}$ and $\varepsilon_i\in\{0,1\}$, and $\mathbf{t}$ is a non-zero vector of length $e$ with entries in $\mathbb{Z}_p$.
\end{definition}
\begin{definition}
The \textbf{size} of a span program is defined to be $d$, the number of rows in $M$. 
\end{definition}
Given a span program $(M,\phi, \mathbf{t})$, the function it computes is given according to the following rule. Given an input string $z$ of $n$ bits, if the vector $\mathbf{t}$ is in $\text{span}(\{r_i: \exists j, \phi(r_i)=(j,z_j)\})$, then output 1. 
Otherwise, output 0. To unpack this, we understand $\phi(r_i)=(j,\varepsilon_i)$ as saying that row $r_i$ maps to some index, $j$, which labels a bit in the input string $z$. If that bit $z_j$ is equal to $\varepsilon_i$, we include that row. 
Repeating this for all rows, we check if the target vector $\mathbf{t}$ is in the span. 

Every function can be computed by a sufficiently large span program \cite{karchmer1993span}.
As a simple example, the AND function is computed by a span program over $\mathbb{Z}_2$ with matrix $M=((1,0),(0,1))$, map $\phi$ such that $\phi(r_1)=(1,1)$ and $\phi(r_2)=(2,1)$, and target vector $\mathbf{t}=(1,1)$. Another simple example is an OR function, computed by $M=((1),(1))$, the map $\phi(r_1)=(1,1)$ and $\phi(r_2)=(2,1)$, and target vector $\mathbf{t}=(1)$. 

A span program is said to be \emph{monotone} if it has $\varepsilon_i=1$ always. This ensures that changing bit values in $z$ from $0$ to $1$ always adds to the set of rows whose span we are checking, so that monotone span programs always compute monotone functions. Conversely, every monotone function can be computed by a monotone span program \cite{karchmer1993span}, as is easy to verify.  

It will be helpful to introduce some notation dealing with span programs. For a given input $z$, the map $\phi$ picks out some of the rows of $M$, whose span will then be checked to see if it includes the target vector. The subset of rows picked out we will denote by $\phi^{-1}(z)$, and refer to as the \emph{activated} rows. The matrix formed from the activated rows we denote $M_{\phi^{-1}(z)}$. The minimal size of a span program over $\mathbb{Z}_p$ computing a function $f$ is denoted $SP_p(f)$. 

\section{Proof of \cref{lemma:decomposition}}\label{appendix:decompositionlemma}

We are now ready to prove \cref{lemma:decomposition}, which we repeat below for convenience. 

\vspace{0.1cm}
\noindent \textbf{\Cref{lemma:decomposition}}\emph{
\,Given a function $f:\{0,1\}^{m}\rightarrow \{0,1\}$, there exist functions 
\begin{align}
    f' &: \{0,1\}^{m+1} \rightarrow \{0,1\}, \nonumber \\
    f_I &:\{0,1\}^{2m+1}\rightarrow \{0,1\}, \nonumber \\
    g &:\{0,1\}^{m+1} \rightarrow \{0,1\}^{2m+1}, \nonumber 
\end{align}  
such that 
\begin{itemize*}
    \item $f'(z,1) = f(z)$
    \item $f'(z,b) = f_I\circ g(z,b)$
    \item $f_I$ is a valid indicator function
    \item $g$ acts on the first $m$ bits of its input by copying each bit $z_i$ and negating one copy, $z_i\rightarrow (z_i,\neg z_i)$. It leaves the final bit $b$ unchanged.
    \item $mSP_p(f_I)\leq SP_p(f)+1$, where $SP_p(h)$ denotes the minimal size of a span program over $\mathbb{Z}_p$ computing $h$, and $mSP_p(h)$ the size of a monotone span program computing $h$.
\end{itemize*}}

\begin{proof}\,
Given $f$, find the minimal sized span program over $\mathbb{Z}_p$ that computes $f$, and label it $(M_f,\phi_f,\mathbf{t}_f)$.
Label the rows of $M_f$ by $r_i$. 
Then, add one row and one column to $M_f$ to define a new matrix $M_{f'}$ with dimensions $(d+1)\times (e+1)$.
Label the rows of $M_f'$ as $r_i'$. 
Set $(M_{f'})_{d+1,e+1}=1$ and otherwise the added row and column entries are set to be $0$.
Extend $\phi_f$ to a new function $\phi_{f'}$ such that $\phi_{f}(r_i)=\phi_{f'}(r_i')$ for all $i\leq d$, and $\phi_{f'}(r_{d+1}') = (m+1,1)$.
Finally, let $\mathbf{t}_{f'} = (\mathbf{t}_f,1)$.
Then $(M_{f'},\phi_{f'},\mathbf{t}_{f'})$ defines a new function $f'$, given by $f'(z,b)=f(z)\wedge b$, so in particular $f'(z,1)=f(z)$. 

Next, we decompose $f'$ into $f_I$ and $g$. Define $g_k(z_k):\{0,1\}^{1}\rightarrow \{0,1\}^{2}$ according to
\begin{align}
    g_k(z_k) = (z_k,\neg z_k)
\end{align}
Then define $g$ by having $g_k$ act on each of the first $m$ bits of the input, producing a string of length $2m+1$. The function $f_I$ is now defined by modifying the span program $(M_{f'},\phi_{f'},\mathbf{t}_{f'})$ to take the output of $g$ as input.
First, the new span program has the same  matrix and target vector as before: $M_I = M_{f'}$ and $\mathbf{t}_I = \mathbf{t}_{f'}$.
Second, define $\phi_I$ by having it map $r_i'$ to the same input bit as $\phi_{f'}$ when $\epsilon_{f',i}=1$, and to the negated copy of that input bit when $\epsilon_{f',i}=0$.
Set $\epsilon_{I,i}=1$ always.
This ensures that $f_I$ and the span program computing it are monotone, but $f'=f_I\circ g$.
Additionally, every $(z,b)$ value which has $f_I(z,b)=1$ must have $b=1$, so $f_I$ is also no-cloning. 
Since secret sharing schemes can be built for any function that is no-cloning and monotone \cite{gottesman2000theory,smith2000quantum}, $f_I$ is a valid indicator function. 
Finally, notice that the monotone span program computing $f_I$ is the same size as the (non-monotone) span program computing $f'$, which in turn has one extra row as compared to the program for $f$. 
\end{proof}

We conclude with an example. Consider the function $f(x,y)=x \oplus y$. A (non-monotone) span program for this function has matrix
\begin{align}
    \begin{pmatrix}
1 & 0 \\
0 & 1 \\
1 & 0 \\
0 & 1 
\end{pmatrix}.
\end{align}
The map $\phi$ is defined by $\phi(r_1)=(1,1)$, $\phi(r_2)=(1,0)$, $\phi(r_3)=(2,1)$, $\phi(r_4)=(2,0)$, and the target vector is $(1,\ 1)$.
It is easy to check cases to confirm this computes $x\oplus y$.

We decompose this in the manner described in \cref{lemma:decomposition}. First, add one column and one row to the matrix according to
\begin{align}
    \begin{pmatrix}
1 & 0 & 0\\
0 & 1 & 0\\
1 & 0 & 0\\
0 & 1 & 0 \\
0 & 0 & 1
\end{pmatrix}.
\end{align}
We add one bit to the inputs, extend the map $\phi$ according to $\phi(r_5) = (3,1)$, and append a 1 to the target vector.
This span program defines the function $f'(x,y,b) = (x\oplus y)\wedge b$. Finally the map $g$ is defined according to
\begin{align}
    g(x,y,b) = (x,\neg x, y,\neg y, b)
\end{align}
and $f_I$ is defined by a span program with the above matrix and map $\phi_I$ defined by $\phi_I(r_1)=(1,1)$, $\phi_I(r_2)=(2,1)$, $\phi_I(r_3)=(3,1)$, $\phi_I(r_4)=(4,1)$, $\phi_I(r_5) = (5,1)$. 

\bibliographystyle{unsrtnat}
\bibliography{biblio}

\end{document}